\documentclass[11pt,letterpaper]{article}

\usepackage[numbers]{natbib}

\usepackage[margin=1in]{geometry}

\usepackage{amsmath,amssymb}
\usepackage{amsthm}
\usepackage[ruled,vlined]{algorithm2e}
\usepackage{pgf, tikz}

\usepackage{lmodern}
\usepackage[T1]{fontenc}
\usepackage{textcomp}
\usepackage[utf8]{inputenc}

\usepackage{nicefrac}

\newtheorem{theorem}{Theorem}
\newtheorem{proposition}[theorem]{Proposition}
\newtheorem{lemma}[theorem]{Lemma}

\newtheorem{claim}[theorem]{Claim}
\newtheorem{claim*}[theorem]{Claim}

\newcommand{\RR}{\ensuremath{\mathbb{R}}}

\newcommand{\e}{e}

\renewcommand{\Pr}[1]{\mbox{\rm\bf Pr}\left[#1\right]}
\newcommand{\Ex}[1]{\mbox{\rm\bf E}\left[#1\right]}

\usepackage{todonotes}

\allowdisplaybreaks

\newcommand{\OPT}{\mathrm{OPT}}
\newcommand{\ALG}{\mathrm{ALG}}
\newcommand{\growingmid}{\mathrel{}\middle|\mathrel{}}

\author{Thomas Kesselheim\thanks{Max-Planck-Institut f\"ur Informatik and Saarland University, Saarbr\"ucken, Germany. \texttt{thomas.kesselheim@mpi-inf.mpg.de}. Supported in part by the DFG through Cluster of Excellence MMCI.} \and Andreas T\"onnis\thanks{Department of Computer Science, RWTH Aachen University, Germany. \texttt{toennis@cs.rwth-aachen.de}. Supported by the DFG GRK/1298 ``AlgoSyn''.}}
\title{Submodular Secretary Problems:\\ Cardinality, Matching, and Linear Constraints}

\begin{document}
\maketitle

\thispagestyle{empty}
\setcounter{page}{0}
\begin{abstract}
We study various generalizations of the secretary problem with submodular objective functions. Generally, a set of requests is revealed step-by-step to an algorithm in random order. For each request, one option has to be selected so as to maximize a monotone submodular function while ensuring feasibility. For our results, we assume that we are given an offline algorithm computing an $\alpha$-approximation for the respective problem. This way, we separate computational limitations from the ones due to the online nature. When only focusing on the online aspect, we can assume $\alpha = 1$.

In the \emph{submodular secretary problem}, feasibility constraints are cardinality constraints, or equivalently, sets are feasible if and only if they are independent sets of a $k$-uniform matroid. That is, out of a randomly ordered stream of entities, one has to select a subset size $k$. For this problem, we present a $0.31\alpha$-competitive algorithm for all $k$, which asymptotically reaches competitive ratio $\nicefrac{\alpha}{e}$ for large $k$. In \emph{submodular secretary matching}, one side of a bipartite graph is revealed online. Upon arrival, each node has to be matched permanently to an offline node or discarded irrevocably. We give an $\frac{\alpha}{4}$-competitive algorithm. This also covers the problem, in which sets of entities are feasible if and only if they are independent with respect to a transversal matroid. In both cases, we improve over previously best known competitive ratios, using a generalization of the algorithm for the classic secretary problem.

Furthermore, we give an $O(\alpha d^{-\frac{2}{B-1}})$-competitive algorithm for submodular function maximization subject to linear packing constraints. Here, $d$ is the column sparsity, that is the maximal number of none-zero entries in a column of the constraint matrix, and $B$ is the minimal capacity of the constraints. Notably, this bound is independent of the total number of constraints. We improve the algorithm to be $O(\alpha d^{-\frac{1}{B-1}})$-competitive if both $d$ and $B$ are known to the algorithm beforehand.

\end{abstract}
\clearpage

\section{Introduction}

In the classic secretary problem, one is presented a sequence of items with different scores online in random order. Upon arrival of an item, one has to decide immediately and irrevocably whether to accept or to reject the current item. The objective is to accept the best of these items. Recently, combinatorial generalizations of this problem have attracted attention. In these settings, feasibility of solutions are stated in terms of matroid or linear constraints. In most cases, these combinatorial generalizations consider linear objective functions. This way, the profit gained by the decision in one step is independent of the other steps.

In this paper, we consider general monotone submodular functions\footnote{A function $f\colon 2^U\rightarrow\RR$ for given ground set $U$ is called \emph{submodular} %
if for all $S\subseteq T\subseteq U$ and every $x\in U\backslash T$ holds $f(S\cup\{x\}) - f(S) \geq f(T\cup\{x\}) - f(T)$.}. For example, the \emph{submodular secretary problem}, independently introduced by Bateni et al.~\cite{DBLP:journals/talg/BateniHZ13} and Gupta et al.~\cite{DBLP:conf/wine/GuptaRST10}, is an online variant of monotone submodular maximization subject to cardinality constraints.  In this problem, we are allowed to select up to $k$ items from a set of $n$ items. The value of a set is represented by a monotone, submodular function. Now, stated as an online problem, items arrive one after the other and every item can only be selected right at the moment when it arrives. The values of the submodular function are only known on subsets of the items that have already arrived. The objective function is designed by an adversary, but the order of the items is uniformly at random.

We call an algorithm $c$-competitive if for any objective function $v$ chosen by the adversary, the set of selected items $\ALG$ satisfies $\Ex{v(\ALG)} \geq (c - o(1)) \cdot v(\OPT)$, where $\OPT$ is a feasible (offline) solution that maximizes $v$ and the $o(1)$-term is asymptotical with respect to the length of the sequence $n$. Note that any algorithm can pretend $n$ to be larger by adding dummy elements at random positions. Therefore, it is safe to assume that $n$ is large compared to $k$.

Previous algorithms for submodular secretary problems were designed by modifying offline approximation algorithms for submodular objectives so that they could be used in the online environment \cite{DBLP:journals/talg/BateniHZ13,DBLP:conf/approx/FeldmanNS11,DBLP:journals/mst/Ma0W16}. In this paper, we take a different approach. Our algorithms are inspired by algorithms for linear objective functions~\cite{DBLP:conf/esa/KesselheimRTV13,DBLP:conf/stoc/KesselheimTRV14}. We repeatedly solve the respective offline optimization problem and use this outcome as a guide to make decisions in the current round. Generally, it is enough to only compute approximate solutions to these offline problems. Our results nicely separate the loss due to the online nature and due to limited computational power. Using polynomial-time computations and existing offline algorithms, we significantly outperform existing online algorithms. Certain submodular functions or kinds of constraints allow better approximations, which immediately transfer to even better competitive ratios. This is, for example, true for submodular maximization subject to a cardinality constraint if the number of allowed items is constant. Also, if computational complexity is no concern like in classical competitive analysis, our competitive ratios become even better.

\subsection{Our Contribution}

Given an $\alpha$-approximate algorithm for monotone submodular maximization subject to a cardinality constraint, we present an $\frac{\alpha}{e}\left(1-\frac{\sqrt{k-1}}{(k+1)\sqrt{2\pi}}\right)$-competitive algorithm for the submodular secretary problem. That is, we achieve a competitive ratio of at least $0.31\alpha$ for any $k \geq 2$. Asymptotically for large $k$, we reach $\nicefrac{\alpha}{e}$.

Our algorithm follows the following natural paradigm. We reject the first $\nicefrac{n}{e}$ items. Afterwards, for each arriving item, we solve the offline optimization problem of the instance that we have seen so far. If the current item is included in this solution and we have not yet accepted too many items, we accept it. Otherwise, we reject it. For the analysis, we bound the expected value obtained by the algorithm recursively. It then remains to solve the recursion and to bound the resulting term. Generally, the recursive approach can be used for any secretary problems with cardinality constraints. It could be of independent interest, especially because it allows to obtain very good bounds also for rather small values of $k$.

One option for the black-box offline algorithm is the standard greedy algorithm by Nemhauser and Wolsey~\cite{NemhauserW78}. It always picks the item of maximum marginal increase until it has picked $k$ items. Generally, this algorithm is $1 - \frac{1}{e}$-approximate. However, it is known that if one compares to the best solution with only $k’ \leq k$ items the approximation factor improves to $1 - \exp\big(-\frac{k}{k’}\big)$. We exploit this fact to give a better analysis of our online algorithm when using the greedy algorithm in each step. We show that the algorithm is 0.238-competitive for any $k$ and asymptotically for large $k$ it is 0.275-competitive.

Additionally, we consider the \emph{submodular secretary matching problem}. In this problem, one side of a bipartite graphs arrives online in random order. Upon arrival, vertices are either matched to a free vertex on the offline side or rejected. The objective is a submodular function on the set of matched pairs or edges. It is easy to see that the submodular secretary problem is a special case of this more general problem. Fortunately, similar algorithmic ideas work here as well. Again, we combine a sampling phase with a black box for the offline problem and get an $\nicefrac{\alpha}{4}$-competitive algorithm. Notably, the analysis turns out to be much simpler compared to the submodular secretary algorithm.

Finally, we show how our new analysis technique can be used to generalize previous results on linear packing programs towards submodular maximization with packing constraints. Here, we use a typical continuous extension towards the expectation on the submodular objective. We parameterize our results in $d$, the column sparsity of the constraint matrix, and $B$, the minimal capacity of the constraints. We achieve a competitive ratio of $\Omega(\alpha d^{-\frac{2}{B-1}})$ if both parameters are not known to the algorithm. If $d$ and $B$ are known beforehand we give different algorithm that is $\Omega(\alpha d^{-\frac{1}{B-1}})$-competitive.

\subsection{Related Work}
\label{sec:related-work}

Although the secretary itself dates back to the 1960s, combinatorial generalizations only gained considerable interest within the last 10 years. One of the earliest combinatorial generalizations and probably the most famous one is the \emph{matroid secretary problem}, introduced by Babaioff et al.~\cite{DBLP:conf/soda/BabaioffIK07}. Here, one has to pick a set of items from a randomly ordered sequence that is an independent set of a matroid. The objective is to maximize the sum of weights of all items picked. It is still believed that there is an $\Omega(1)$-competitive algorithm for this problem; the currently best known algorithms achieve a competitive ratio of $\Omega(\nicefrac{1}{\log\log(\rho)})$ for matroids of rank $\rho$~\cite{DBLP:conf/soda/FeldmanSZ15, DBLP:conf/focs/Lachish14}. Additionally, there are constant competitive algorithms known for many special cases, e.g., for transversal matroids there is an $\nicefrac{1}{e}$-competitive algorithm~\cite{DBLP:conf/esa/KesselheimRTV13} and for $k$-uniform matroids there is an $1-O(\nicefrac{1}{\sqrt{k}})$-competitive algorithm~\cite{DBLP:conf/soda/Kleinberg05}. Both are known to be optimal. Other examples include graphical matroids, for which there is a $\nicefrac{1}{2e}$-competitive algorithm~\cite{DBLP:conf/icalp/KorulaP09}, and laminar matroids, for which a $\nicefrac{1}{9.6}$-competitive algorithm is known~\cite{DBLP:journals/mst/Ma0W16}. Further well-studied generalizations feature \emph{linear constraints}. This includes online packing LPs \cite{DBLP:conf/sigecom/DevenurH09,DBLP:journals/mor/MolinaroR14,DBLP:journals/ior/AgrawalWY14,DBLP:conf/stoc/KesselheimTRV14} and online edge-weighted matching \cite{DBLP:conf/esa/KesselheimRTV13, DBLP:conf/icalp/KorulaP09}, for which optimal algorithms are known. Also the online variant of the generalized assignment problem~\cite{DBLP:conf/stoc/KesselheimTRV14} %
has been studied.

All these secretary problems have in common that the objective function is linear. Compared to other objective functions this has the clear advantage that the gain due to a choice in one round is independent of choices in other rounds. Interdependencies between the rounds only arise due to the constraints. Bateni et al.~\cite{DBLP:journals/talg/BateniHZ13} and Gupta et al.~\cite{DBLP:conf/wine/GuptaRST10} independently started work on submodular objective functions in the secretary setting. To this point, the best known results are a $\frac{e-1}{e^2+e}\approx 0.170$-competitive algorithm for $k$-uniform matroids~\cite{DBLP:conf/approx/FeldmanNS11} and a $\frac{1}{95}$-competitive algorithm for submodular secretary matching~\cite{DBLP:journals/mst/Ma0W16}. In case there are $m$ linear packing constraints, the best known algorithm is $O(m)$-competitive~\cite{DBLP:journals/talg/BateniHZ13}. For matroid constraints, Feldman and Zenklusen~\cite{DBLP:conf/focs/FeldmanZ15} give a reduction, turning a $c$-competitive algorithm for linear objective functions to an $\Omega(c^2)$-competitive one for linear objective functions. Furthermore, they give the first $\Omega(\nicefrac{1}{\log\log\rho})$-competitive algorithm for the submodular matroid secretary problem. Feldman and Izsak~\cite{DBLP:journals/corr/FeldmanI15} consider more general objective functions, which are not necessarily submodular. They give competitive algorithms for cardinality constraint secretary problems that are parameterized in the \emph{supermodular degree} of the objective function. 

Agrawal and Devanur~\cite{DBLP:conf/soda/AgrawalD15} study concave constraints and concave objective functions. These results, however, do not generalize submodular objectives because they require the dimension of the vector space to be low. Representing an arbitrary submodular function would require the dimension to be as large as $n$. Another related problem is submodular welfare maximization. In this case, even the greedy algorithm is known to be $\nicefrac{1}{2}$-competitive in adversarial order but at least $0.505$-competitive in random order~\cite{DBLP:conf/stoc/KorulaMZ15}.

In the offline setting, submodular function maximization is computationally hard if the function is given through a value oracle. There are efficient algorithms that approximate a monotone, submodular function over a matroid or under a knapsack-constraint with a factor of $(1-\nicefrac{1}{e})$~\cite{DBLP:journals/siamcomp/CalinescuCPV11, DBLP:journals/orl/Sviridenko04}. As a special case the generalized assignment problem can also be efficiently approximated up to a factor of $(1-\nicefrac{1}{e})$~\cite{DBLP:journals/siamcomp/CalinescuCPV11}.
For a constant number of linear constraints, there is also a $(1-\epsilon)(1-\nicefrac{1}{e})$-approximation algorithm~\cite{DBLP:journals/mor/KulikST13}. In the non-monotone domain, there is an algorithm for cardinality constraint submodular maximization with an approximation factor in the range $[\nicefrac{1}{e} + 0.004, \nicefrac{1}{2}]$ depending on the cardinality~\cite{DBLP:conf/soda/BuchbinderFNS14}.

\section{Submodular Secretary Problem}
\label{sec:submodular_k-secretary}

Let us first turn to the submodular secretary problem. Here, a set of items from a universe $U$, $\lvert U \rvert = n$, is presented to the algorithm in random order. For each arriving $j \in U$, the algorithm has to decide whether to accept or to reject it, being allowed to accept up to $k$ items in total. The objective is to maximize a monotone submodular function $v\colon 2^U \to \RR_{\geq 0}$. This function is defined by an adversary and known to the algorithm only restricted to the subsets of items that have already arrived. This problem extends the secretary problem for $k$-uniform matroids with linear objective functions, which was solved by Kleinberg~\cite{DBLP:conf/soda/Kleinberg05}. The previously best known competitive factor is $\frac{e-1}{e^2+e}\approx 0.170$ ~\cite{DBLP:conf/approx/FeldmanNS11}.

Depending on the kind of the submodular function and its representation, the corresponding offline optimization problem (monotone submodular maximization with cardinality constraint) can be computationally hard. In order to be able to focus on the online nature of the problem, we assume that we are given an offline algorithm $\mathcal{A}$ that for any $L \subseteq U$ returns an $\alpha$-approximation of the best solution within $L$. Formally, $v(\mathcal{A}(L)) \geq \alpha \max_{T \subseteq L, \lvert T \vert \leq k} v(T)$. Note that $\mathcal{A}$ is allowed to exploit any additional structure of the function $v$. For different $L$ and $L'$, $\mathcal{A}(L)$ and $\mathcal{A}(L')$ do not have to be consistent, but the output $\mathcal{A}(L)$ must be identical, irrespective of the arrival order on $L$.

Our online algorithm, Algorithm~\ref{alg:k-secretary}, uses algorithm $\mathcal{A}$ as a subroutine as follows. It starts by rejecting the first $pn$ items. For every following item $j$, it runs $\mathcal{A}(L)$, where $L$ is the set of items that have arrived up to this point. If $j \in \mathcal{A}(L)$ we call $j$ tentatively selected. Furthermore if the set of accepted items $S$ contains less than $k$ items and $j$ is tentatively selected, then the algorithm adds $j$ to $S$. Otherwise, it rejects $j$.

\begin{algorithm}
\caption{Submodular $k$-secretary}\label{alg:k-secretary}
Drop the first $\lceil pn \rceil - 1$ items\;
\For(\tcp*[f]{online steps $\ell = \lceil pn \rceil$ to $n$}){item $j$ arriving in round $\ell \geq \lceil pn \rceil$}{
  Set $U^{\leq\ell} := U^{\leq\ell-1} \cup \{j\}$\; 
  Let $S^{(\ell)} = \mathcal{A}(U^{\leq\ell})$;\tcp*[f]{black box $\alpha$-approximation}\\
  \If(\tcp*[f]{tentative allocation}){$j \in S^{(\ell)}$}{
    \If(\tcp*[f]{feasibility test}){$\lvert \mathrm{Accepted} \vert < k$}{
      Add $j$ to $\mathrm{Accepted}$; \tcp*[f]{online allocation}\\
    }
  }
}
\end{algorithm}

\begin{theorem}\label{theorem:k-secretary}
Algorithm~\ref{alg:k-secretary} for the submodular secretary problem is $\frac{\alpha}{e} \left(1 - \frac{\sqrt{k-1}}{(k+1)\sqrt{2 \pi}} \right)$-competitive with sample size $pn = \frac{n}{e}$.
\end{theorem}

\subsection{Analysis Technique}\label{subsec:technique}

Before proving Theorem~\ref{theorem:k-secretary}, let us shed some light on the way we lower-bound the value of the submodular objective function. To this end, we consider the expected value of the set of all tentatively selected items $T$. In other words, we pretend all selections our algorithm tries to make are actually feasible. It seems natural to bound the expected value of $T$ by adding up the marginal gains round-by-round given the tentative selections in earlier rounds. Unfortunately, this introduces complicated dependencies on the order of arrival of previous items. Therefore, we take a different approach and bound the respective marginal gains with respect of tentative selections in \emph{future} rounds. The important insight is that this keeps the dependencies manageable.

\begin{proposition}\label{prop:tentative_value}
The set of all items $T$ that are tentatively selected by Algorithm~\ref{alg:k-secretary} has an expected value of $\Ex{v(T)} \geq \left(\frac{\alpha}{e}-\frac{\alpha}{n}\right)\cdot v(\OPT)$ if the algorithm is run with sample size $pn = \frac{n}{e}$.
\end{proposition}

\begin{proof}
Let $T^{\geq\ell}$ denote the set of tentatively selected items that arrive in or after round $\ell$. Formally, we have $T^{\geq\ell} = \{j\} \cup T^{\geq\ell+1}$ if $j \in \mathcal{A}(U^{\leq\ell})$ and $T^{\geq\ell} = T^{\geq\ell+1}$ otherwise.

We consider a different random process to define the $T^{\geq \ell}$ random variables, which results in the same distribution. First, we draw one item from $U$ uniformly to come last. This determines the value of $T^{\geq n}$. Then we continue by drawing on item out of the remaining ones to come second to last, determining $T^{\geq n - 1}$. Generally, this means that conditioning on $U^{\leq\ell}$ and the values of $T^{\geq \ell'}$, for $\ell' > \ell$, the item $j$ is drawn uniformly at random from $U^{\leq\ell}$ and the respective outcome determines $T^{\geq \ell}$.

We bound the expected tentative value collected in rounds $\ell$ to $n$ conditioned on the items that arrived before round $\ell$ and conditioned on all items that are tentatively selected afterwards
\begin{align*}
& \Ex{v(T^{\geq\ell})\growingmid U^{\leq\ell}, T^{\geq\ell'} \text{for all }\ell' > \ell} = \frac{1}{\ell}v\left(\mathcal{A}(U^{\leq\ell})\growingmid T^{\geq \ell+1}\right) + v(T^{\geq\ell+1})\\
&\geq \frac{1}{\ell}v\left(\mathcal{A}(U^{\leq\ell}) \cup T^{\geq\ell+1} \right) - \frac{1}{\ell} v(T^{\geq\ell+1}) + v(T^{\geq\ell+1}) \geq \frac{1}{\ell}v\left(\mathcal{A}(U^{\leq\ell})\right) + \left(1-\frac{1}{\ell}\right)v(T^{\geq\ell+1})\enspace.
\end{align*}
We take the expectation over the remaining randomization and get a simple recursion
\[\Ex{v(T^{\geq\ell})} \geq \frac{1}{\ell}\Ex{v\left(\mathcal{A}(U^{\leq\ell})\right)} + \left(1-\frac{1}{\ell}\right)\Ex{v(T^{\geq\ell+1})}\enspace.\]

Observe that $\OPT \cap U^{\leq \ell}$ is fully contained in $U^{\leq \ell}$ and has size at most $k$. Therefore, the approximation guarantee of $\mathcal{A}$ yields that $v(\mathcal{A}(U^{\leq \ell})) \geq \alpha v(\OPT \cap U^{\leq \ell})$. Furthermore, submodularity gives us $\Ex{v(\OPT \cap U^{\leq \ell})} \geq \frac{\ell}{n} v(\OPT)$ because each item is included in $U^{\leq \ell}$ with probability $\frac{\ell}{n}$. In combination, this gives us
\begin{equation}\label{eq:tentative_value}
\Ex{v(\mathcal{A}(U^{\leq \ell}))} \geq \alpha\Ex{v(\OPT \cap U^{\leq \ell})} \geq \alpha \frac{\ell}{n} v(\OPT)\enspace.
\end{equation}

Now we solve the recursion
\[
\Ex{v(T^{\geq\ell})} \geq \frac{\alpha}{n}v(\OPT) + \left(1-\frac{1}{\ell}\right)\Ex{v(T^{\geq\ell+1})}
= \sum_{j=\ell}^n\prod_{i=\ell}^j\left(1-\frac{1}{i}\right) \frac{\alpha}{n}v(\OPT)\enspace.
\]
We have $\prod_{i=\ell}^{j-1} \left(1-\frac{1}{i}\right) = \frac{\ell-1}{j-1}$ and $\sum_{j=\ell}^n \frac{1}{j-1} \geq \ln(\frac{n}{\ell})$ for all $\ell \geq 2$. This yields
\[
\Ex{v(T^{\geq\ell})} \geq \sum_{j=\ell}^n\prod_{i=\ell}^j\left(1-\frac{1}{i}\right) \frac{\alpha}{n}v(\OPT) = \frac{\alpha}{n}v(\OPT)\sum_{j=\ell}^n \frac{\ell-1}{j-1} \geq \frac{\ell-1}{n}\ln\left(\frac{n}{\ell}\right)\alpha v(\OPT)\enspace.
\]

With $\ell = pn$ and sample size $pn =\frac{n}{e}$, we get
\[
\Ex{v(T^{\geq pn})} \geq \frac{pn-1}{n}\ln\left(\frac{1}{p}\right)\alpha v(\OPT) = \left(\frac{1}{e} - \frac{1}{n}\right)\alpha v(\OPT)\enspace. \qedhere
\]
\end{proof}

The probability of a tentative selection in round $\ell$ is $\frac{k}{\ell}$. This means, in expectation, we make $\sum_{\ell = \frac{n}{e}}^n \frac{k}{\ell} \approx k$ tentative selections. Therefore, for large values of $k$, it is likely that most tentative selections are feasible. This way, we could already derive guarantee for large $k$. However, for small $k$, the derived bound would be far to pessimistic. This is due to the fact that we bound the marginal gain of an item based on all \emph{tentative} future ones. If some of them are indeed not feasible, we underestimate the contribution of earlier items. Therefore, Theorem~\ref{theorem:k-secretary} requires a more involved recursion that is based on the idea from this section, but also incorporates the probability that an item is feasible directly.

\subsection{Proof of Theorem~\ref{theorem:k-secretary}}\label{subsec:k-secretary}

To prove the theorem, we will derive a lower bound on the value collected by the algorithm starting from an arbitrary round $\ell \in [n]$ with an arbitrary remaining capacity $r \in \{0, 1, \ldots, k\}$. The random variables $\ALG^{\geq \ell}_r \subseteq U$ represent the set of first $r$ items that a hypothetical run of the algorithm would collect if it started the \emph{for} loop of Algorithm~\ref{alg:k-secretary} in round $\ell$. Formally, we define them recursively as follows. We set $\ALG^{\geq \ell}_0 = \emptyset$ for all $\ell$ and $\ALG^{\geq n+1}_r = \emptyset$ for all $r$. For $\ell \in [n]$, $r > 0$, let $j$ be the item arriving in round $\ell$, and $U^{\leq\ell}$ be the set of items arriving until and including round $\ell$. We define $\ALG^{\geq \ell}_r = \{ j \} \cup \ALG^{\geq \ell + 1}_{r - 1}$ if $j \in \mathcal{A}(U^{\leq\ell})$ and $\ALG^{\geq \ell}_r = \ALG^{\geq \ell + 1}_r$ otherwise.
Note that by this definition $\ALG = \ALG^{\geq pn}_k$. Furthermore, for every possible arrival order, $\ALG^{\geq \ell}_r$ is pointwise a superset of $\ALG^{\geq \ell}_{r - 1}$ for $r > 0$.

In Lemma~\ref{lemma:with-recursion}, we show a recursive lower bound on the value of these sets. In this part, the precise definition of $\ALG^{\geq \ell}_r$ will be crucial to avoid complex dependencies. Afterwards, in Lemma~\ref{lemma:without-recursion}, we solve this recursion. Given this solution, we can finally prove Theorem~\ref{theorem:k-secretary}.

\begin{lemma}
\label{lemma:with-recursion}
For all $\ell \in [n]$ and $r \in \{0, 1, \ldots, k\}$, we have
\[
\Ex{v(\ALG^{\geq \ell}_r)} \geq \frac{1}{\ell} \left( \Ex{v(\mathcal{A}(U^{\leq\ell}))} + (k - 1) \Ex{v(\ALG^{\geq \ell + 1}_{r - 1})} + (\ell - k) \Ex{v(\ALG^{\geq \ell + 1}_r)} \right) \enspace.
\]
\end{lemma}

\begin{proof}
Like explained in Section~\ref{subsec:technique}, we first draw one item from $U$ uniformly at random to be the item that arrives in round $n$. This defines the values of $\ALG^{\geq n}_r$ for all $r$. Then we draw another item to be the second to last one and so on. In this way, we can condition on $U^{\leq\ell}$ and the values of $\ALG^{\geq \ell'}_r$, for $\ell' > \ell$ and all $r$. In round $\ell$, the item $j$ is drawn uniformly at random from $U^{\leq\ell}$ and the respective outcome determines $\ALG^{\geq \ell}_r$ for all $r$.
This allows us to write for $r > 0$
\begin{align*}
& \Ex{v(\ALG^{\geq \ell}_r) \growingmid U^{\leq\ell}, \ALG^{\geq \ell'}_{r'} \text{ for all $\ell' > \ell$ and all $r'$}} \\
& = \frac{1}{\ell} \left( \sum_{j \in \mathcal{A}(U^{\leq\ell})} v( \{ j \} \cup \ALG^{\geq \ell + 1}_{r - 1}) + \lvert U^{\leq\ell} \setminus \mathcal{A}(U^{\leq\ell}) \rvert v(\ALG^{\geq \ell + 1}_r) \right) \enspace.
\end{align*}

By submodularity, we have
\[
\sum_{j \in \mathcal{A}(U^{\leq\ell})} \left( v( \{ j \} \cup \ALG^{\geq \ell + 1}_{r - 1}) - v(\ALG^{\geq \ell + 1}_{r - 1}) \right) \geq v(\mathcal{A}(U^{\leq\ell}) \cup \ALG^{\geq \ell + 1}_{r - 1}) - v(\ALG^{\geq \ell + 1}_{r - 1}) \enspace,
\]
and hence
\[
\sum_{j \in \mathcal{A}(U^{\leq\ell})} v( \{ j \} \cup \ALG^{\geq \ell + 1}_{r - 1}) \geq v(\mathcal{A}(U^{\leq\ell}) \cup \ALG^{\geq \ell + 1}_{r - 1}) + (\lvert \mathcal{A}(U^{\leq\ell}) \rvert - 1) v(\ALG^{\geq \ell + 1}_{r - 1}) \enspace.
\]
This gives us
\begin{align*}
& \Ex{v(\ALG^{\geq \ell}_r) \growingmid U^{\leq\ell}, \ALG^{\geq \ell'}_{r'} \text{ for all $\ell' > \ell$ and all $r'$}} \\
& \geq \frac{1}{\ell} \left( v(\mathcal{A}(U^{\leq\ell}) \cup \ALG^{\geq \ell + 1}_{r - 1}) + (\lvert \mathcal{A}(U^{\leq\ell}) \rvert - 1) v(\ALG^{\geq \ell + 1}_{r - 1}) + \lvert U^{\leq\ell} \setminus \mathcal{A}(U^{\leq\ell}) \rvert v(\ALG^{\geq \ell + 1}_r) \right) \enspace.
\end{align*}
Furthermore, by applying the monotonicity of $v$ and the facts that $\lvert \mathcal{A}(U^{\leq\ell}) \rvert \leq k$ and $\ALG^{\geq \ell + 1}_{r - 1} \subseteq \ALG^{\geq \ell + 1}_r$, we get
\begin{align*}
& \Ex{v(\ALG^{\geq \ell}_r) \growingmid U^{\leq\ell}, \ALG^{\geq \ell'}_{r'} \text{ for all $\ell' > \ell$ and all $r'$}} \\
& \geq \frac{1}{\ell} \left( v(\mathcal{A}(U^{\leq\ell})) + (k - 1) v(\ALG^{\geq \ell + 1}_{r - 1}) + (\ell - k) v(\ALG^{\geq \ell + 1}_r) \right) \enspace.
\end{align*}
Taking the expectation over all remaining randomization yields
\[
\Ex{v(\ALG^{\geq \ell}_r)} \geq \frac{1}{\ell} \left( \Ex{v(\mathcal{A}(U^{\leq\ell}))} + (k - 1) \Ex{v(\ALG^{\geq \ell + 1}_{r - 1})} + (\ell - k) \Ex{v(\ALG^{\geq \ell + 1}_r)} \right) \enspace. \qedhere
\]
\end{proof}

The next step is to solve the recursion.
\begin{lemma}
\label{lemma:without-recursion}
For all $\ell \in [n]$, $\ell \geq k^2 + k$, and $r \in \{0, 1, \ldots, k\}$, we have
\begin{equation}
\Ex{v(\ALG^{\geq \ell}_r)} \geq \left(\frac{r\ell}{(k-1)n} - \frac{1}{k-1}\left(\frac{\ell}{n}\right)^k\sum_{r'=0}^{r-1}\sum_{i=0}^{r'}\frac{(k-1)^i}{i!}\ln^i\left(\frac{n}{\ell}\right) - \frac{3 k^2r}{(k-1)n} \right) \alpha v(\OPT) \enspace.
\label{eq:without-recursion-simpler}
\end{equation}
\end{lemma}

\begin{proof}
As a first step, we eliminate the recursive reference from $\ALG^{\geq \ell}_r$ to $\ALG^{\geq \ell + 1}_r$. To this end, we count the rounds until the next item is accepted. Repeatedly inserting the bound for $\ALG^{\geq \ell + 1}_r$ into the one for $\ALG^{\geq \ell}_r$ gives us
\[
\Ex{v(\ALG^{\geq \ell}_r)}\geq\sum_{j=\ell}^n\left(\prod_{i=\ell}^{j-1}\left(1-\frac{k}{i}\right)\left(\frac{k-1}{j}\Ex{v(\ALG^{\geq j+1}_{r-1})} + \frac{1}{j}\Ex{v(\mathcal{A}(U^{\leq j}))}\right)\right) \enspace.
\]
With Equation~\eqref{eq:tentative_value} in Section~\ref{subsec:technique} we have $\Ex{v(\mathcal{A}(U^{\leq j}))} \geq \frac{j}{n} \alpha v(\OPT)$.

We use $\prod_{i=\ell}^{j-1}\left(1-\frac{k}{i}\right) = \frac{(\ell-1)!}{(\ell-k-1)!}\frac{(j-k-1)!}{(j-1)!} \geq \left(\frac{\ell-k}{j-k}\right)^k$ and get
\begin{equation}
\Ex{v(\ALG^{\geq \ell}_r)}\geq\sum_{j=\ell}^n\left(\left(\frac{\ell-k}{j-k}\right)^k\left(\frac{k-1}{j+1}\Ex{v(\ALG^{\geq j+1}_{r-1})} + \frac{\alpha}{n}v(\OPT)\right)\right)\enspace.\label{eq:hookforgreedyanalysis}
\end{equation}

To shown that \eqref{eq:without-recursion-simpler} provides a lower bound on the functions defined by this recursion, we perform an induction on $r$. Note that Equation~\eqref{eq:without-recursion-simpler} trivially holds for $r = 0$. In order to prove it holds for a given $r > 0$, we assume that  it is fulfilled for $r-1$ for all $\ell \in [n]$. From this, we will conclude that Equation~\eqref{eq:without-recursion-simpler} also holds for $r$ for all $\ell \in [n]$. To show that \eqref{eq:hookforgreedyanalysis} is solved by \eqref{eq:without-recursion-simpler}, we use the induction hypothesis and plug in the bound for $\Ex{v(\ALG^{\geq j+1}_{r-1})}$. This gives us

\begin{align*}
\frac{\Ex{v(\ALG^{\geq \ell}_r)}}{\alpha v(\OPT)} &\geq \sum_{j=\ell}^n\left(\frac{\ell-k}{j-k}\right)^k \frac{k-1}{j+1}\bigg(\frac{(r-1)(j+1)}{(k-1)n} - \frac{3k^2(r-1)}{(k-1)n} + \frac{1}{n}\\
&\qquad - \frac{1}{k-1}\left(\frac{j+1}{n}\right)^k\sum_{r'=0}^{r-2}\sum_{i=0}^{r'}\frac{(k-1)^i}{i!}\ln^i\left(\frac{n}{j+1}\right)\bigg)\\
&= \sum_{j=\ell}^n\left(\frac{\ell-k}{j-k}\right)^k \frac{r}{n} - \sum_{j=\ell}^n\left(\frac{\ell-k}{j-k}\right)^k \frac{1}{j+1}\left(\frac{j+1}{n}\right)^k\sum_{r'=0}^{r-2}\sum_{i=0}^{r'}\frac{(k-1)^i}{i!}\ln^i\left(\frac{n}{j+1}\right)\\
&\qquad - \sum_{j=\ell}^n\left(\frac{\ell-k}{j-k}\right)^k \frac{3k^2(r-1)}{(j+1)n}\enspace.
\end{align*}
In the negative terms, we bound $\frac{\ell-k}{j-k}\leq \frac{\ell}{j}$ and use $\left(\frac{j+1}{j}\right)^k\leq e^{\frac{k}{j}} \leq e^{\frac{k}{\ell}}\leq 1 + 2 \frac{k}{\ell}$. Finally in the last sum, we bound $\frac{1}{j+1}\leq\frac{1}{\ell}$ once

\begin{align*}
\frac{\Ex{v(\ALG^{\geq \ell}_r)}}{\alpha v(\OPT)} &\geq \sum_{j=\ell}^n\left(\frac{\ell-k}{j-k}\right)^k \frac{r}{n} - \left(\frac{\ell}{n}\right)^k \sum_{j=\ell}^n \frac{\left(1 + 2 \frac{k}{\ell}\right)}{j+1}\sum_{r'=0}^{r-2}\sum_{i=0}^{r'}\frac{(k-1)^i}{i!}\ln^i\left(\frac{n}{j+1}\right)\\
&\qquad - \sum_{j=\ell}^n\left(\frac{\ell}{j}\right)^k \frac{3k^2(r-1)}{\ell n}\enspace.
\end{align*}

We approximate both sums over $j$ through integrals by using
\[
\sum_{j=\ell}^n\frac{1}{(j-k)^k} \geq \int_\ell^n \frac{1}{(j-k)^k} dj = \frac{1}{k-1}\left(\frac{1}{(\ell-k)^{k-1}} - \frac{1}{(n-k)^{k-1}}\right)
\]
and
\[\sum_{j=\ell}^n\frac{\ln^i(\nicefrac{n}{(j+1)})}{j+1} \leq \int_{\ell-1}^{n-1} \frac{\ln^i(\nicefrac{n}{(j+1)})}{j+1} dj = \left[-\frac{\ln^{i+1}(\nicefrac{n}{(j+1)})}{i+1}\right]_{\ell-1}^{n-1} = \frac{\ln^{i+1}(\nicefrac{n}{\ell})}{i+1}\enspace.\]
This yields

\begin{align*}
\frac{\Ex{v(\ALG^{\geq \ell}_r)}}{\alpha v(\OPT)} &\geq \frac{r(\ell-k)}{(k-1)n} \left(1-\left(\frac{\ell-k}{n-k}\right)^{k-1}\right) - \left(\frac{\ell}{n}\right)^k \left(1 + 2 \frac{k}{\ell}\right)\sum_{r'=0}^{r-2}\sum_{i=0}^{r'}\frac{(k-1)^i}{i!}\frac{\ln^{i+1}\left(\frac{n}{\ell}\right)}{i+1}\\
&\qquad - \frac{3k^2(r-1)}{(k-1) n}\left(1-\left(\frac{\ell}{n}\right)^{k-1}\right)\enspace.
\end{align*}

We perform an index shift in the inner sum and propagate the shift to the outer sum
\begin{align*}
\sum_{r'=0}^{r-2}\sum_{i=0}^{r'}\frac{(k-1)^i}{i!}\frac{\ln(\nicefrac{n}{\ell})^{i+1}}{i+1} &= \frac{1}{k-1}\sum_{r'=0}^{r-2}\sum_{i=1}^{r'+1}\frac{(k-1)^i}{i!} \ln^i\left(\frac{n}{\ell}\right)\\
&= \frac{1}{k-1}\sum_{r'=1}^{r-1}\sum_{i=1}^{r'}\frac{(k-1)^i}{i!}\ln^i\left(\frac{n}{\ell}\right) \\
&= \frac{1}{k-1}\sum_{r'=0}^{r-1}\sum_{i=0}^{r'}\frac{(k-1)^i}{i!}\ln^i\left(\frac{n}{\ell}\right) - \frac{r}{k-1} \enspace .
\end{align*}
Now we solve the brackets and use the term split off in the index shift to simplify the expression. We get

\begin{align*}
\frac{\Ex{v(\ALG^{\geq \ell}_r)}}{\alpha v(\OPT)} &\geq \frac{r(\ell-k)}{(k-1)n} - \frac{r(\ell-k)}{(k-1)n}\left(\frac{\ell-k}{n-k}\right)^{k-1} + \left(\frac{\ell}{n}\right)^k \frac{\left(1 + 2 \frac{k}{\ell}\right)}{k-1}r\\
&\qquad - \left(\frac{\ell}{n}\right)^k \frac{\left(1 + 2 \frac{k}{\ell}\right)}{k-1}
\sum_{r'=0}^{r-1}\sum_{i=0}^{r'}\frac{(k-1)^i}{i!}\ln^i\left(\frac{n}{\ell}\right) - \frac{3k^2(r-1)}{(k-1)n}\\
&\geq \frac{r\ell}{(k-1)n}  - \frac{rk}{(k-1)n}
- \left(\frac{\ell}{n}\right)^k \frac{\left(1 + 2 \frac{k}{\ell}\right)}{k-1}\sum_{r'=0}^{r-1}\sum_{i=0}^{r'}\frac{(k-1)^i}{i!}\ln^i\left(\frac{n}{\ell}\right)\\
&\qquad - \frac{3k^2(r-1)}{(k-1)n}\enspace.
\end{align*}

At this point, we only have to show that the following inequality holds
\begin{align*}
\frac{rk}{(k-1)n} + \left(\frac{\ell}{n}\right)^k \frac{2 \frac{k}{\ell}}{k-1}\sum_{r'=0}^{r-1}\sum_{i=0}^{r'}\frac{(k-1)^i}{i!}\ln^i\left(\frac{n}{\ell}\right) + \frac{3k^2(r-1)}{(k-1)n} \leq \frac{3 k^2 r}{(k-1)n}\enspace.
\end{align*}
We bound the inner sum with the corresponding exponential function
\[
\sum_{i=0}^{r'}\frac{(k-1)^i}{i!}\ln^i\left(\frac{n}{\ell}\right) \leq \sum_{i=0}^{\infty}\frac{(k-1)^i}{i!}\ln^i\left(\frac{n}{\ell}\right) = \exp\left((k-1)\ln\left(\frac{n}{\ell}\right)\right) = \left(\frac{n}{\ell}\right)^{k-1}\enspace.
\]
This term is independent of $r'$. We eliminate the sum over $r'$ and get
\[
\frac{rk}{(k-1)n} + \frac{\ell}{n} \frac{r2 \frac{k}{\ell}}{k-1} = \frac{3kr}{(k-1)n} \leq \frac{3k^2}{(k-1)n}\enspace. \qedhere
\]
\end{proof}

\begin{proof}[Proof of Theorem~\ref{theorem:k-secretary}]
To complete the proof of the theorem, we apply Lemma~\ref{lemma:without-recursion} for $\ell = pn$ and $r = k$. This gives us $\Ex{v(\ALG)} = \Ex{v(\ALG^{\geq pn}_k)}$ and thus
\[
\Ex{v(\ALG)} \geq \left(\frac{pk}{k-1} - \frac{1}{k-1} p^k \sum_{r'=0}^{k-1} \sum_{i=0}^{r'} \frac{(k-1)^i}{i!}\ln^i \left(\frac{1}{p} \right) - \frac{6 k^2}{n} \right) \cdot \alpha v(\OPT)\enspace.
\]
For $p = \frac{1}{e}$, we have $\ln\left( \frac{1}{p} \right) = 1$. This allows us to reorder the occurring double sum as follows
\begin{align*}
\sum_{r'=0}^{k-1} \sum_{i=0}^{r'} \frac{(k-1)^i}{i!} &= \sum_{i=0}^{k-1} (k-i) \frac{(k-1)^i}{i!}\\
&= k\sum_{i=0}^{k-1} \frac{(k-1)^i}{i!} - (k-1)\sum_{i=1}^{k-1} \frac{(k-1)^{i-1}}{(i-1)!}\\
&= \sum_{i=0}^{k-1} \frac{(k-1)^i}{i!} + \frac{(k-1)^k}{(k-1)!}\enspace.
\end{align*}
By definition of the exponential function $e^x = \sum_{i=0}^{\infty} \frac{x^i}{i!}$. For $x>0$, all terms of the infinite sum are positive. This yields $e^x \geq \sum_{i=0}^{k-1} \frac{x^i}{i!} + \frac{x^k}{k!} + \frac{x^{k+1}}{(k+1)!}$ and thus by setting $x = k-1$ we get
\[
\sum_{r'=0}^{k-1} \sum_{i=0}^{r'} \frac{(k-1)^i}{i!} \leq e^{k-1} - \frac{(k-1)^k}{k!} - \frac{(k-1)^{k+1}}{(k+1)!} + \frac{(k-1)^k}{(k-1)!}\enspace.
\]
This implies
\begin{align*}
\frac{\Ex{v(\ALG)}}{\alpha v(\OPT)} &\geq \frac{k}{e(k-1)} - \frac{1}{e^k(k-1)} \left( e^{k-1} - \frac{(k-1)^k}{k!} - \frac{(k-1)^{k+1}}{(k+1)!} + \frac{(k-1)^k}{(k-1)!} \right) - \frac{6 k^2}{n} \\
&= \frac{1}{e} + \frac{1}{e^k} \frac{(k-1)^{k-1}}{k!} + \frac{1}{e^k} \frac{(k-1)^k}{(k+1)!} - \frac{1}{e^k} \frac{(k-1)^{k-1}}{(k-1)!} - \frac{6 k^2}{n} \\
&= \frac{1}{e} - \frac{1}{e^k} \frac{k-1}{k+1} \frac{(k-1)^{k-1}}{(k-1)!} - \frac{6 k^2}{n} \enspace.
\end{align*}
It only remains to apply the Stirling approximation $(k-1)! \geq \sqrt{2\pi (k-1)}\left(\frac{k-1}{e}\right)^{k-1}$ to get
\[
\frac{\Ex{v(\ALG)}}{\alpha v(\OPT)} \geq \frac{1}{e}\left(1 - \frac{\sqrt{k-1}}{(k+1)\sqrt{2\pi}}\right) - \frac{6 k^2}{n} \enspace.\qedhere
\]
\end{proof}

\subsection{Improved Analysis for the Greedy Algorithm}
One possible choice for the algorithm $\mathcal{A}$ is the greedy algorithm by Nemhauser and Wolsey~\cite{NemhauserW78}. It repeatedly picks the item with the highest marginal increase compared to the items chosen so far until $k$ items have been picked. As pointed out in \cite{Krause12survey}, the approximation guarantee would improve further when picking more items according to the greedy rule. In other words, if we let our algorithm pick $k$ elements but compare the outcome to the optimal solution of only $k'$ items, the approximation factor improves to $1 - \exp\left( - \nicefrac{k}{k'} \right)$.

We can exploit this fact in the analysis of the online algorithm that uses the greedy algorithm as $\mathcal{A}$ in Algorithm~\ref{alg:k-secretary}. The reason is that in early rounds only some items of the optimal solution have arrived. Our algorithm, however, always chooses a set of size $k$ for $S^{(\ell)} = \mathcal{A}(U^{\leq\ell})$. In the generic analysis, we show that $\Ex{v(\mathcal{A}(U^{\leq\ell}))} \geq \alpha \frac{\ell}{n} v(\OPT)$. In case of $\mathcal{A}$ being the greedy algorithm, we can improve this bound as follows.

\begin{lemma}
\label{lemma:alpha_ell}
$\Ex{v(\mathcal{A}(U^{\leq\ell}))} \geq \alpha_\ell \frac{\ell}{n} v(\OPT)$ for $\alpha_\ell = 1 - \frac{\ell}{e n} - \frac{1}{ek}$.
\end{lemma}

\begin{proof}
Consider the offline optimum $\OPT$ and $\OPT \cap U^{\leq\ell}$, its restriction to the items that arrive by round $\ell$. Let $Z = \lvert \OPT \cap U^{\leq\ell} \rvert$ be the number of $\OPT$ items that arrive by round $\ell$.

Condition on any value of $Z$. Observe that by symmetry the probability of every $\OPT$ item to have arrived by round $\ell$ is $\frac{Z}{k}$. Therefore, submodularity implies $\Ex{ v(\OPT \cap U^{\leq\ell}) \growingmid Z} \geq \frac{Z}{k} v(\OPT)$. Letting the greedy algorithm pick $k$ elements, it achieves value at least $\left(1 - \exp\left( - \frac{k}{Z} \right) \right) v(\OPT \cap U^{\leq\ell})$. In combination, this gives us
\[
\Ex{v(\mathcal{A}(U^{\leq\ell})) \growingmid Z} \geq \left(1 - \exp\left( - \frac{k}{Z} \right) \right) \frac{Z}{k} v(\OPT) \enspace.
\]
We now use the fact that $\exp\left( \frac{k}{Z} \right) \geq e \frac{k}{Z}$ because $Z \leq k$. Therefore $\exp\left( - \frac{k}{Z} \right) \leq \frac{Z}{ek}$ and
\[
\Ex{v(\mathcal{A}(U^{\leq\ell})) \growingmid Z} \geq \left(1 - \frac{Z}{ek} \right) \frac{Z}{k} v(\OPT) \enspace.
\]

It remains to take the expectation over $Z$. We have $\Ex{Z} = \frac{\ell}{n} k$ and $\Ex{Z^2} \leq \frac{\ell}{n} k + \left( \frac{\ell}{n} k \right)^2$. This implies
\[
\Ex{v(\mathcal{A}(U^{\leq\ell}))} \geq \left(\frac{\Ex{Z}}{k} - \frac{\Ex{Z^2}}{ek^2}\right) v(\OPT) \geq \left( \frac{\ell}{n} - \frac{\ell^2}{e n^2}  - \frac{\ell}{ekn} \right) v(\OPT) \enspace. \qedhere
\]
\end{proof}

Given this lemma, we can follow similar steps as in the proof of Theorem~\ref{theorem:k-secretary} to show an improved guarantee of this particular algorithm. In more detail, we get competitive ratios of at least 0.177 for any $k \geq 2$. Asymptotically for large $k$ we reach 0.275.
\begin{theorem}
\label{theorem:greedy}
Algorithm~\ref{alg:k-secretary} using the greedy algorithm for $\mathcal{A}$ is $\frac{1 + \frac{1}{2 e^3} - \frac{3}{2e} - \frac{e-1}{e^2k}}{e - 1} \left(1 - \frac{\sqrt{k-1}}{(k+1)\sqrt{2 \pi}} \right)$-competitive with sample size $pn = \frac{n}{e}$.
\end{theorem}

To prove Theorem~\ref{theorem:greedy}, we combine Lemmas~\ref{lemma:with-recursion} and \ref{lemma:alpha_ell}, which give us a recursive formula for $\ALG^{\geq\ell}_r$. This time, the recursion is more complex. Therefore, our proof strategy is to first write $v(\ALG^{\geq\ell}_r)$ as the following kind of linear combination (Claim~\ref{claim:greedy:removerecursion})
\[
\Ex{v(\ALG^{\geq\ell}_r)} \geq \sum_{j = \ell}^n t_{\ell, j} \frac{\alpha_j v(\OPT)}{n} \enspace.
\]
Then we show that the occurring coefficients $t_{\ell, j}$ are non-increasing (Claim~\ref{claim:greedy:nonincreasing}) for fixed $\ell$. As both $t_{\ell, j} \geq t_{\ell, j+1}$ and $\alpha_j \geq \alpha_{j+1}$, this then allows to apply Chebyshev's sum inequality to get
\[
\Ex{v(\ALG^{\geq\ell}_r)} \geq \left( \frac{1}{n - \ell + 1} \sum_{j = \ell}^n \alpha_j \right) \left( \sum_{j = \ell}^n t_{\ell, j} \frac{v(\OPT)}{n} \right) \enspace.
\]
This means that we get the same kind of bound as in Section~\ref{subsec:k-secretary} but $\alpha$ is effectively replaced by the average of the involved $\alpha_j$, rather than their minimum.

\begin{claim}
\label{claim:greedy:removerecursion}
Lemma~\ref{lemma:with-recursion} implies
\[
\Ex{v(\ALG^{\geq\ell}_r)} \geq \sum_{j = \ell}^n \frac{a_{\ell, j-1}}{j} \Ex{v(\mathcal{A}(U^{\leq\ell}))} \sum_{r' = 0}^{r-1} \sum_{\substack{M \subseteq \{ \ell, \ldots, j-1 \}\\ \lvert M \rvert = r'}} \left( \prod_{i \in M} \frac{k-1}{i} \right)
\]
with $a_{\ell, j-1} = \prod_{i = \ell}^{j-1} \left( 1 - \frac{k}{i} \right)$.
\end{claim}

\begin{proof}
We perform an induction on $\ell$. Assume that the claim has been shown for all $r$ for $\ell + 1$. In Lemma~\ref{lemma:with-recursion}, we have shown
\[
\Ex{v(\ALG^{\geq \ell}_r)} \geq \frac{1}{\ell} \left( \Ex{v(\mathcal{A}(U^{\leq\ell}))} + (k - 1) \Ex{v(\ALG^{\geq \ell + 1}_{r - 1})} + (\ell - k) \Ex{v(\ALG^{\geq \ell + 1}_r)} \right) \enspace,
\]

Now we use the induction hypothesis
\begin{align*}
\Ex{v(\ALG^{\geq \ell}_r)} & \geq \frac{1}{\ell} \Ex{v(\mathcal{A}(U^{\leq\ell}))} \\
& \quad + \frac{k - 1}{\ell} \sum_{j = \ell+1}^n \frac{a_{\ell+1, j-1}}{j} \Ex{v(\mathcal{A}(U^{\leq\ell}))} \sum_{r' = 0}^{r-2} \sum_{\substack{M \subseteq \{ \ell+1, \ldots, j-1 \}\\ \lvert M \rvert = r'}} \left( \prod_{i \in M} \frac{k-1}{i} \right) \\
& \quad + \frac{\ell - k}{\ell} \sum_{j = \ell+1}^n \frac{a_{\ell+1, j-1}}{j} \Ex{v(\mathcal{A}(U^{\leq\ell}))} \sum_{r' = 0}^{r-1} \sum_{\substack{M \subseteq \{ \ell+1, \ldots, j-1 \}\\ \lvert M \rvert = r'}} \left( \prod_{i \in M} \frac{k-1}{i} \right)\enspace.
\end{align*}
We perform an index shift, use $\frac{\ell-k}{\ell}a_{\ell+1, j-1} = a_{\ell, j-1}$ and get
\begin{align*}
\Ex{v(\ALG^{\geq \ell}_r)} & = \frac{a_{\ell, \ell-1}}{\ell} \Ex{v(\mathcal{A}(U^{\leq\ell}))} \\
& \quad + \sum_{j = \ell+1}^n \frac{a_{\ell+1, j-1}}{j} \Ex{v(\mathcal{A}(U^{\leq\ell}))} \sum_{r' = 1}^{r-1}  \frac{k - 1}{\ell} \sum_{\substack{M \subseteq \{ \ell+1, \ldots, j-1 \}\\ \lvert M \rvert = r' - 1}} \left( \prod_{i \in M} \frac{k-1}{i} \right) \\
& \quad + \sum_{j = \ell+1}^n \frac{a_{\ell, j-1}}{j} \Ex{v(\mathcal{A}(U^{\leq\ell}))} \sum_{r' = 0}^{r-1} \sum_{\substack{M \subseteq \{ \ell+1, \ldots, j-1 \}\\ \lvert M \rvert = r'}} \left( \prod_{i \in M} \frac{k-1}{i} \right)\enspace.
\end{align*}
We have $\frac{k-1}{\ell} \geq \frac{k-1}{i}$ for all $i \geq \ell$ and therefore we can merge the factor for the current round into the product. In a sense, the $\frac{k-1}{\ell}$ factor stands for choosing an item in the current round, and it gets worse if we chose one in a future round instead.
Additionally we use $a_{\ell+1, j-1} \geq a_{\ell, j-1}$ and omit the second large sum entirely.

For the final equality, we use the fact that $\sum_{r' = 0}^{r-1} \sum_{\substack{M \subseteq \emptyset, \lvert M \rvert = r'}} \left( \prod_{i \in M} \frac{k-1}{i} \right) = 1$ because the inner sum is empty for all $r' > 0$
\begin{align*}
\Ex{v(\ALG^{\geq \ell}_r)} & \geq \frac{a_{\ell, \ell-1}}{\ell} \Ex{v(\mathcal{A}(U^{\leq\ell}))} \\
& \quad + \sum_{j = \ell+1}^n \frac{a_{\ell, j-1}}{j} \Ex{v(\mathcal{A}(U^{\leq\ell}))} \sum_{r' = 0}^{r-1} \sum_{\substack{M \subseteq \{ \ell, \ldots, j-1 \}\\ \lvert M \rvert = r'}} \left( \prod_{i \in M} \frac{k-1}{i} \right) \\
& = \sum_{j = \ell}^n \frac{a_{\ell, j-1}}{j} \Ex{v(\mathcal{A}(U^{\leq\ell}))} \sum_{r' = 0}^{r-1} \sum_{\substack{M \subseteq \{ \ell, \ldots, j-1 \}\\ \lvert M \rvert = r'}} \left( \prod_{i \in M} \frac{k-1}{i} \right) \enspace. \qedhere
\end{align*}
\end{proof}
 
\begin{claim}
\label{claim:greedy:nonincreasing}
Let
\[
t_{\ell, j} = a_{\ell, j-1} \sum_{r' = 0}^{r-1} \sum_{\substack{M \subseteq \{ \ell, \ldots, j-1 \}\\ \lvert M \rvert = r'}} \left( \prod_{i \in M} \frac{k-1}{i} \right), \qquad \text{where} \quad a_{\ell, j-1} = \prod_{i = \ell}^{j-1} \left( 1 - \frac{k}{i} \right)
\]
For fixed $\ell$, the sequence $t_{\ell, j}$ is non-increasing in $j$.
\end{claim}

\begin{proof}
We will show that $t_{\ell, j+1} \leq \beta_j t_{\ell, j}$ for some $\beta_j \leq 1$. To this end, we consider the definition of $t_{\ell, j+1}$ and split of a double sum that contains all terms where $j\in M$. In those terms, we know that $j$ is selected and therefore the factor $\frac{k-1}{j}$ always exists in the product. We get
\begin{align*}
t_{\ell, j+1} & = a_{\ell, j} \sum_{r' = 0}^{r-1} \sum_{\substack{M \subseteq \{ \ell, \ldots, j \}\\ \lvert M \rvert = r'}} \left( \prod_{i \in M} \frac{k-1}{i} \right) \\
& = a_{\ell, j} \left( \sum_{r' = 0}^{r-1} \sum_{\substack{M \subseteq \{ \ell, \ldots, j-1 \}\\ \lvert M \rvert = r'}} \left( \prod_{i \in M} \frac{k-1}{i} \right) + \frac{k-1}{j} \sum_{r' = 0}^{r-1} \sum_{\substack{M \subseteq \{ \ell, \ldots, j-1 \}\\ \lvert M \rvert = r' - 1}} \left( \prod_{i \in M} \frac{k-1}{i} \right) \right)\enspace.
\end{align*}
Both double sums are nearly identical. We fill up the missing terms in the smaller one and bound by the following expression. Finally, we replace the remaining double sum with the definition of $t_{\ell, j}$
\begin{align*}
t_{\ell, j+1} & \leq a_{\ell, j} \left( 1 + \frac{k-1}{j} \right) \sum_{r' = 0}^{r-1} \sum_{\substack{M \subseteq \{ \ell, \ldots, j-1 \}\\ \lvert M \rvert = r'}} \left( \prod_{i \in M} \frac{k-1}{i} \right) \\
& = \frac{a_{\ell, j}}{a_{\ell, j-1}} \left( 1 + \frac{k-1}{j} \right) t_{\ell, j} \enspace.
\end{align*}
As we have $\frac{a_{\ell, j}}{a_{\ell, j-1}} \left( 1 + \frac{k-1}{j} \right) = \left( 1 + \frac{k-1}{j} \right) \left( 1 - \frac{k}{j} \right) = 1 - \frac{k}{j} + \frac{k-1}{j} - \frac{k(k-1)}{j^2} \leq 1$, the claim follows.
\end{proof}

\begin{proof}[Proof of Theorem~\ref{theorem:greedy}]
Now we can proceed to the proof of Theorem~\ref{theorem:greedy}. So far, we have shown that
\[
\Ex{v(\ALG^{\geq\ell}_r)} \geq \sum_{j = \ell}^n \frac{t_{\ell, j}}{j} \Ex{v(\mathcal{A}(U^{\leq\ell}))} \quad \text{ for } \quad
t_{\ell, j} = a_{\ell, j-1} \sum_{r' = 0}^{r-1} \sum_{\substack{M \subseteq \{ \ell, \ldots, j-1 \}\\ \lvert M \rvert = r'}} \left( \prod_{i \in M} \frac{k-1}{i} \right)
\]
with $a_{\ell, j-1} = \prod_{i = \ell}^{j-1} \left( 1 - \frac{k}{i} \right)$. Furthermore, Lemma~\ref{lemma:alpha_ell} shows that $\frac{\Ex{v(\mathcal{A}(U^{\leq\ell}))}}{j} \geq \frac{\alpha_j v(\OPT)}{n}$ for $\alpha_\ell = 1 - \frac{\ell}{e n} - \frac{1}{ek}$.

As both $t_{\ell, j}$ and $\alpha_j$ are non-increasing in $j$, we can use Chebyshev's sum inequality to get
\[
\Ex{v(\ALG^{\geq\ell}_r)} \geq \sum_{j = \ell}^n t_{\ell, j} \frac{\alpha_j v(\OPT)}{n} \geq \left( \sum_{j = \ell}^n t_{\ell, j} \frac{v(\OPT)}{n} \right) \left( \frac{1}{n - \ell} \sum_{j=\ell}^n \alpha_j \right)
\]
It now remains to bound these two terms.

First we show that the sum $\sum_{j = \ell}^n t_{\ell, j}\frac{c}{n}$ with $c = v(\OPT)$ is lower-bounded by a recursion of the form of Equation~\eqref{eq:hookforgreedyanalysis}. Calculations like in Lemma~\ref{lemma:without-recursion} will then give us the respective bound. Similar the previous proof, we use $a_{\ell, j-1} = \prod_{i = \ell}^{j-1} \left( 1 - \frac{k}{i} \right) \geq \left(\frac{\ell-k}{j-k}\right)^k$ and get
\begin{align*} 
\sum_{j = \ell}^n t_{\ell, j} \frac{v(\OPT)}{n} &= \sum_{j = \ell}^na_{\ell, j-1} \sum_{r' = 0}^{r-1} \sum_{\substack{M \subseteq \{ \ell, \ldots, j \}\\\lvert M \rvert = r'}} \left( \prod_{i \in M} \frac{k-1}{i} \right)\frac{c}{n}\\
&\geq \sum_{j = \ell}^n \left(\frac{\ell-k}{j-k}\right)^k \sum_{r' = 0}^{r-1} \sum_{\substack{M \subseteq \{ \ell, \ldots, j \}\\\lvert M \rvert = r'}} \left( \prod_{i \in M} \frac{k-1}{i+1} \right)\frac{c}{n} \enspace.
\end{align*}

Let now
\[
b_{\ell, r'} = \sum_{j = \ell}^n \left(\frac{\ell-k}{j-k}\right)^k \sum_{r' = 0}^{r-1} \sum_{\substack{M \subseteq \{ \ell, \ldots, j \}\\\lvert M \rvert = r'}} \left( \prod_{i \in M} \frac{k-1}{i+1} \right)\frac{c}{n}
\]
We combine the two inner sums and then pull out the earliest element $m\in M \subseteq \{ \ell, \ldots, j \}$ recursively. We move the corresponding factor out of the product and get
\begin{align*}
b_{\ell, r'} & = \sum_{j = \ell}^n \left( \frac{\ell-k}{j-k} \right)^k \sum_{\substack{M \subseteq \{ \ell, \ldots, j \}\\\lvert M \rvert \leq r'}} \left(\prod_{i \in M} \frac{k-1}{i+1}\right)\frac{c}{n} \\
& = \sum_{j = \ell}^n \left( \frac{\ell-k}{j-k} \right)^k \left(\frac{c}{n} +\sum_{m=\ell}^{j-1} \frac{k-1}{m+1} \sum_{\substack{M \subseteq \{ m+1, \ldots, j \}\\\lvert M \rvert \leq r'-1}} \left(\prod_{i \in M} \frac{k-1}{i+1} \right)\frac{c}{n}\right)\enspace.
\end{align*}
At this point, we change the order of summation such that we sum over $m$ first. We can keep the constant part in place, since both sums $\sum_{j=\ell}^n\left(\frac{\ell-k}{j-k}\right)^k = \sum_{m=\ell}^n\left(\frac{\ell-k}{m-k}\right)^k$ amount the same. Now the inner part matches the recursion given above
\begin{align*}
b_{\ell, r'} & = \sum_{m=\ell}^n \left( \frac{\ell-k}{m-k} \right)^k \left( \frac{c}{n} + \frac{k-1}{m+1} \sum_{j = m+1}^n \left( \frac{m-k}{j-k} \right)^k \sum_{\substack{M \subseteq \{ m+1, \ldots, j \}\\\lvert M \rvert \leq r'-1}} \left(\prod_{i \in M} \frac{k-1}{i}\right)\frac{c}{n} \right)\\
& = \sum_{m=\ell}^n \left( \frac{\ell-k}{m-k} \right)^k \left( \frac{c}{n} + \frac{k-1}{m+1} b_{m+1, r'-1} \right)\enspace.
\end{align*}

From this point on, we follow the proof of Lemma~\ref{lemma:without-recursion} and get the following lemma.

\begin{lemma}
Given a recursion of the form
\[
b_{\ell, r} = \sum_{j=\ell}^n\left(\left(\frac{\ell-k}{j-k}\right)^k\left(\frac{k-1}{j+1} b_{j+1, r-1} + \frac{c}{n} \right)\right)
\]
with $b_{n+1, r} = 0$ and $b_{\ell, 0} = 0$. Then
\[
b_{\ell, r} \geq \left(\frac{r(\ell - k)}{(k-1)n} - \frac{1}{k-1}\left(\frac{\ell - k}{n-k}\right)^k\sum_{r'=0}^{r-1}\sum_{i=0}^{r'}\frac{(k-1)^i}{i!}\ln^i\left(\frac{n}{\ell}\right) - \frac{3k^2r}{(k-1)n} \right) c\enspace.
\]
\end{lemma}

Consequently, following the calculations in the proof of Theorem~\ref{theorem:k-secretary} 
\[
\Ex{v(\ALG)} = \Ex{v(\ALG^{\geq \nicefrac{n}{e}}_k} \geq \frac{1}{e} \left(1 - \frac{\sqrt{k-1}}{(k+1)\sqrt{2 \pi}} - \frac{6ek^2}{n} \right) \left( \frac{1}{n - \nicefrac{n}{e}} \sum_{j=\nicefrac{n}{e}}^n \alpha_j \right)v(\OPT)
\]
For $\alpha_j = 1 - \frac{j}{en} - \frac{1}{ek}$, we can bound the last term by
\begin{align*}
\frac{1}{n - \nicefrac{n}{e}} \sum_{j = n/e}^n \left( 1 - \frac{j}{en}  - \frac{1}{ek} \right) & \geq \frac{1}{n - \nicefrac{n}{e}} \int_{n/e}^n \left( 1 - \frac{j}{en} - \frac{1}{ek} \right) dj\\
& = \frac{1}{n - \nicefrac{n}{e}} \int_{1/e}^1 \left( 1 - \frac{x}{e} - \frac{1}{ek} \right) n dx\\
& = \frac{1}{1 - \nicefrac{1}{e}} \int_{1/e}^1 \left( 1 - \frac{x}{e} - \frac{1}{ek} \right) dx\\
& = \frac{1}{1 - \nicefrac{1}{e}} \left( 1 + \frac{1}{2e^3} - \frac{3}{2e} - \frac{e-1}{e^2k} \right)\enspace.
\end{align*}
For large $k$, we have an asymptotic competitive ratio of $\frac{1}{e}\left( 1 + \frac{1}{2e^3} - \frac{3}{2e}\right) \approx 0.275$.
\end{proof}

\section{Submodular Matching}\label{sec:matching}
Next, we consider the online submodular bipartite matching problem. In the offline version, we are given a bipartite graph $G = (L \cup R, E)$ and a monotone, submodular, non-decreasing objective function $v\colon 2^E \rightarrow \RR_{\geq 0}$. The objective is to find a matching $M \subseteq E$ that maximizes $v(M)$. In the online version, the set $L$ arrives online. Once a vertex in $L$ arrives, we get to know its incident edges. At any point in time, we know the values of the objective function only restricted to subsets of the edges incident to the vertices that have already arrived. This problem also generalizes the submodular matroid secretary problem with transversal matroids.

We present an $\frac{\alpha}{4}$-competitive algorithm, where $\alpha$ could be $\frac{1}{3}$ for a simple greedy algorithm~\cite{DBLP:journals/mp/NemhauserWF78}. The best known approximation algorithms are local search algorithms that give a $\frac{1}{2+\epsilon}$-approximation on bipartite matchings~\cite{DBLP:journals/mor/LeeSV10, DBLP:conf/esa/FeldmanNSW11}. The previously best known online algorithm is the simulated greedy algorithm with a competitive ratio of $\nicefrac{1}{95}$~\cite{DBLP:journals/mst/Ma0W16}. 

Algorithm~\ref{alg:submodular-matching} first samples a $\nicefrac{1}{2}$-fraction of the input sequence. Then, whenever a new candidate arrives, it $\alpha$-approximates the optimal matching on the known part of the graph with respect to the submodular objective function. If the current online vertex is matched in this matching and if its matching partner is still available, then we add the pair to the output allocation. This design paradigm has been successfully applied to linear objective functions before \cite{DBLP:conf/esa/KesselheimRTV13}. However, in the submodular case, the individual contribution on an edge to the eventual objective function value depends on what other edges are selected. Using an approach similar to the one in the previous section, we keep dependencies manageable.

\begin{theorem}\label{theorem:matching}
Algorithm~\ref{alg:submodular-matching} is an $\frac{\alpha}{4}$-competitive online algorithm for the submodular secretary matching problem that uses $\frac{n}{2}$ calls to an offline $\alpha$-approximation algorithm for submodular matching.
\end{theorem}

\begin{algorithm}
\caption{Submodular Bipartite Online Matching}\label{alg:submodular-matching}
Drop the first $\lceil \frac{n}{2} \rceil - 1$ vertices\;
\For(\tcp*[f]{online steps $\ell = \lceil pn \rceil$ to $n$}){vertex $u\in L$ in round $\ell \geq \lceil \frac{n}{2} \rceil$}{
  Set $L^{\leq\ell} := L^{\leq\ell-1} \cup \{u\}$\; 
  Let $M^{(\ell)} = \mathcal{A}(L^{\leq\ell} \cup R)$;\tcp*[f]{black box $\alpha$-approximation}\\
    Let $e^{(\ell)}:=(u, r)$ be the edge assigned to $u$ in $M^{(\ell)}$; \tcp*[f]{tentative edge}\\
 \If(\tcp*[f]{feasibility test}){$\mathrm{Accepted} \cup e^{(\ell)}$ \emph{is a matching}}{
      Add $e^{(\ell)}$ to $\mathrm{Accepted}$; \tcp*[f]{online allocation}\\
  }
}
\end{algorithm}

We denote the set of matching edges allocated by the algorithm in rounds $\ell$ to $n$ with $\ALG^{\geq\ell}$ and the set of tentative edges over the same period with $T^{\geq\ell}$. Furthermore let $\hat{e}^{(\ell)}$ be a set containing the tentative edge of round $\ell$ if this edge was actually assigned and empty otherwise. That is, $\hat{e}^{(\ell)} = \{e^{(\ell)}\}$ if $e^{(\ell)}$ is allocated and $\hat{e}^{(\ell)} = \emptyset$ otherwise. Please note that $e^{(\ell)}$ might be empty. For $S, S' \subseteq E$, we denote the contribution of the subset $S$ to $S'$ by $v(S \mid S') = v(S \cup S') - v(S')$.

The proof follows the natural approach described in Section~\ref{subsec:technique}. First we bound the tentative value collected in every round against the future rounds, then we bound the probability that a tentative allocation is feasible.

\begin{lemma}\label{lemma:matching_tentative}
In every round $\ell$ fix the tentative edges that will be selected in the future rounds $\ell+1, \dots, n$. Then the marginal contribution of the tentative edge $e^{(\ell)}$ selected by the online algorithm in round $\ell$ is $\Ex{v\left(\{e^{(\ell)} \} \growingmid \ALG^{\geq\ell+1} \right) \growingmid L^{\leq\ell}, T^{\geq\ell+1}} \geq \frac{1}{\ell} \left(v(\mathcal{A}(L^{\leq\ell})) - v(T^{\geq\ell+1})\right)$.
\end{lemma}

This lemma is shown in a way similar to Proposition~\ref{prop:tentative_value}. To avoid complex dependencies, we will use that $v\left(e^{(\ell)} \growingmid \ALG^{\geq\ell+1} \right) \geq v\left(e^{(\ell)} \growingmid T^{\geq\ell+1}\right)$ because of submodularity of $v$ and since  $\ALG^{\geq\ell+1} \subseteq T^{\geq\ell+1}$. 

\begin{proof}
With $L^{\leq\ell}$ fixed, the algorithm's output $\mathcal{A}(L^{\leq\ell})$ is determined as well. The online vertex in round $\ell$ is as drawn uniformly at random from all vertices in $L^{\leq\ell}$. This gives us
\begin{align*}
\Ex{v\left(\{e^{(\ell)}\} \growingmid T^{\geq\ell+1}\right) \growingmid L^{\leq\ell}, T^{\geq\ell+1}} &\geq \frac{1}{\ell} v\left(\mathcal{A}(L^{\leq\ell}) \growingmid T^{\geq\ell+1}\right) \geq \frac{1}{\ell} \left( v(\mathcal{A}(L^{\leq\ell})) - v(T^{\geq\ell+1})\right)\enspace.\qedhere
\end{align*}
\end{proof}

\begin{lemma}
\label{lemma:matching_collision-probability}
The probability that a tentative edge $e^{(\ell)}$ is feasible given all vertices that arrived earlier $L^{\leq\ell}$ and all future tentative edges $T^{\geq\ell+1}$ is $\Pr{ \mathrm{Accepted} \cup e^{(\ell)} \text{ is a matching} \growingmid L^{\leq\ell}, T^{\geq\ell+1}} \geq \frac{\frac{n}{2} - 1}{\ell - 1}$.
\end{lemma}

This was already shown in \cite{DBLP:conf/esa/KesselheimRTV13}. For completeness, we provide a proof here.

\begin{proof}
First, we consider the probability, that a tentatively selected edge $e^{(\ell)}$ makes it to the final matching. The probability that a tentative edge $e^{(\ell)}$ is feasible is at least $\prod_{j = pn}^{\ell - 1} (1 - \frac{1}{j}) = \frac{pn - 1}{\ell - 1}$. Since in the previous local matchings $M^{(j)}$ for $pn\leq j<\ell$ at most one vertex $i$ is matched to the partner of $\ell$ in $M_\ell$. Vertices arrive in random order, we interpret this as drawing one vertex uniformly at random from all vertices that arrived. Therefore $i$ is drawn uniformly at random from $L^{\leq j}$, thus the probability that $i$ is the current online vertex is $\frac{1}{j}$.

Formally, we have
\[
\Pr{\hat{e}^{(\ell)} \neq \emptyset \growingmid e^{(\ell)}, L^{\leq\ell}, T^{\geq\ell+1}} \geq \frac{pn - 1}{\ell - 1} \enspace.\qedhere
\]
\end{proof}

\begin{proof}[Proof of Theorem~\ref{theorem:matching}]
Combining Lemmas~\ref{lemma:matching_tentative} and \ref{lemma:matching_collision-probability}, we get that in every round $\ell$ for a fixed set $L^{\leq\ell}$ and $T^{\geq\ell+1}$ we have 
\[
\Ex{v\left(\hat{e}^{(\ell)} \growingmid \ALG^{\geq\ell+1} \right) \growingmid L^{\leq\ell}, T^{\geq\ell+1}} 
\geq \frac{1}{\ell} \frac{pn - 1}{\ell - 1} \left(v(\mathcal{A}(L^{\leq\ell} \cup R)) - v(T^{\geq\ell+1})\right)
\]
and therefore
\[
\Ex{v\left(\hat{e}^{(\ell)} \growingmid \ALG^{\geq\ell+1} \right)} \geq \frac{1}{\ell} \frac{pn - 1}{\ell - 1} \left(\Ex{v(\mathcal{A}(L^{\leq\ell} \cup R))} - \Ex{v(T^{\geq\ell+1})}\right)\enspace.
\]
Using Lemma~\ref{lemma:matching_collision-probability} another time, we also have $\Ex{v(\ALG^{\geq\ell+1})} \geq \frac{pn - 1}{\ell - 1} \Ex{v(T^{\geq\ell+1})}$. Furthermore, to bound $\Ex{v(\mathcal{A}(L^{\leq\ell} \cup R))}$, we use that the optimal solution on the subgraph induced by $L^{\leq\ell} \cup R$ is at least as good as the optimal solution restricted to the edges in this subgraph. As every edge appears with probability $\frac{\ell}{n}$, submodularity gives us $\Ex{v(\mathcal{A}(L^{\leq\ell} \cup R))} \geq \alpha \frac{\ell}{n} v(\OPT)$. In combination, this yields
\[
\Ex{v\left(\hat{e}^{(\ell))} \growingmid \ALG^{\geq\ell+1} \right)} \geq \frac{\alpha}{n} \frac{pn - 1}{\ell - 1} v(\OPT) - \frac{1}{\ell} \Ex{v(\ALG^{\geq\ell+1})}\enspace.
\]
As $\ALG^{\geq\ell} = \hat{e}^{(\ell)} \cup \ALG^{\geq\ell+1}$, we get the following recursion
\[
\Ex{v((\ALG^{\geq\ell})} \geq \frac{\alpha}{n} \frac{pn - 1}{\ell - 1} v(\OPT) + \left( 1 - \frac{1}{\ell}\right) \Ex{v(\ALG^{\geq\ell+1})}\enspace.
\]

Now we solve the tail recursion
\[
\Ex{v(\ALG^{\geq\ell})} \geq \sum_{j = \ell}^n \prod_{i=\ell}^{j-1} \left(1-\frac{1}{i}\right)\frac{1}{j-1}\left(p - \frac{1}{n}\right)\alpha v(\OPT) \enspace.
\]
We have $\prod_{i=\ell}^{j-1} \left(1-\frac{1}{i}\right) = \frac{\ell-1}{j-1}$ and $\sum_{j=\ell}^n \frac{1}{(j-1)^2} \geq \frac{1}{\ell} - \frac{1}{n}$ thus we get
\begin{align*}
\Ex{v(\ALG^{\geq\ell})}&\geq \sum_{j = \ell}^n \prod_{i=\ell}^{j-1} \left(1-\frac{1}{i}\right)\frac{1}{j - 1}\left(p - \frac{1}{n}\right)\alpha v(\OPT)\\
&\geq\sum_{j = \ell}^n \frac{\ell - 1}{(j-1)^2} \left(p - \frac{1}{n}\right)\alpha\OPT\\
&\geq \left(\frac{1}{\ell} - \frac{1}{n}\right)(\ell - 1)\left(p - \frac{1}{n}\right)\alpha v(\OPT)\enspace.
\end{align*}

The expected value of the online algorithm $\Ex{v(\ALG^{\geq pn})}$ is maximized for $p = \nicefrac{1}{2}$

\begin{align*}
\frac{\Ex{v(\ALG^{\geq pn})}}{\alpha v(\OPT)} &\geq \left(p-\frac{1}{n}\right)\left(1-\frac{1}{pn}-p + \frac{1}{n}\right)\\
&=\left(p-p^2-O\left(\frac{1}{n}\right)\right)=\left(\frac{1}{4} - O\left(\frac{1}{n}\right)\right)\enspace.\qedhere
\end{align*}
\end{proof}

\section{Submodular Function subject to Linear Packing Constraints}\label{sec:LP}

We now generalize the setting to feature arbitrary linear packing constraints. That is, each item $j$ is associated a variable $y_j$ and there are $m$ constraints of the form $\sum_{j \in U} a_{i, j} y_j \leq b_i$ with $a_{i, j} \geq 0$. The coefficients $a_{i, j}$ are chosen by an adversary and are revealed to the online algorithm once the respective item arrives. Immediately and irrevocably, we have to either accept or reject the item, which corresponds to setting $y_j$ to $0$ or $1$. The best previous result is a constant competitive algorithm for a single constraint and $\Omega(\nicefrac{1}{m})$-competitive for multiple constraints, where $m$ is the number of constraints~\cite{DBLP:journals/talg/BateniHZ13}.

Our algorithms extend the ones presented in \cite{DBLP:conf/stoc/KesselheimTRV14} from linear to submodular objective functions. Again, they rely on a suitable algorithm solving the offline optimization problem. In this case, we need a fractional allocation $x \in [0, 1]^U$, which we evaluate in terms of the multilinear extension $F(x) = \sum_{R\subseteq U} \left(\prod_{i\in R} f(R) x_i \prod_{i \notin R} (1-x_i)\right)$. In more detail, we assume that for any packing polytope $P \subseteq [0, 1]^U$, $F(\mathcal{A}_F(P)) \geq \alpha \sup_{x \in P} F(x)$. For example, the continuous greedy process by Calinescu et al.~\cite{DBLP:journals/siamcomp/CalinescuCPV11} provides a $(1-\nicefrac{1}{e})$-approximation in polynomial time. As the set $P$, we use $\mathcal{P}(\frac{\ell}{n}, S)$, which is defined to be the set of vectors $x \geq 0$, for which $A x \leq \frac{\ell}{n} b$ and $x_i = 0$ if $i \not\in S$. This is the polytope of the solution space with scaled down constraints and restricted on the variables that arrived so far. 

Our bounds are parameterized in the capacity ratio $B$ and the column sparsity $d$. The capacity ratio $B$ is defined by $B = \min_{i\in[m]}\frac{b_i}{\max_{j\in[n]}a_{i, j}}$. The column sparsity $d$ is the maximal number of none-zero entries in a column of the constraint matrix $A$. We consider two variants of this problem, where either the $B$ and $d$ are known to the algorithm or not.

\begin{theorem}\label{theorem:LP_unknown}
There is an $\Omega\left(\alpha d^{-\frac{2}{B-1}}\right)$-competitive online algorithm for submodular maximization subject to linear constraints with unknown capacity ratio $B\geq2$ and unknown column sparsity $d$.
\end{theorem}

If the minimal capacity $B$ and the column sparsity $d$ are known, we can fine-tune Algorithm~\ref{algo:lp} and add a sampling phase that is dependent on those two parameters.

\begin{theorem}\label{theorem:LP_known}
There is an $\Omega\left(\alpha d^{-\frac{1}{B-1}}\right)$-competitive online algorithm for submodular maximization subject to linear constraints with known capacity ratio $B\geq2$ and known column sparsity $d$.
\end{theorem}

Note that, although the algorithm $\mathcal{A}$ returns fractional solutions, the output of our online algorithms is integral. The competitive ratio is between the integral solution of the online algorithm and the optimal fractional allocation with respect to the multilinear extension.

\begin{algorithm}
\caption{Submodular Function Maximization subject to Linear Constraints\label{algo:lp}}
\SetKwInOut{Input}{Input} \SetKwInOut{Output}{Output}
Let $x := 0$ and $S := \emptyset$ be the index set of known requests\;
\For(\tcp*[f]{steps $\ell = 1$ to $n$}){\emph{each arriving request} $j$}{
  Set $S := S \cup \{j\}$ and $\ell := |S|$\;
  Let $\tilde{x}^{(\ell)} := \mathcal{A}_F(\mathcal{P}(\frac{\ell}{n}, S))$; \tcp*[f]{fractional $\alpha$-approximation on scaled polytope}\\
  Set $\hat{x}_j^{(\ell)} = 1$ with probability $\tilde{x}^{(\ell)}_{j}$; \tcp*[f]{tentative allocation after rand.\ rounding}\\
  \If(\tcp*[f]{feasibility test}){$A(x+\hat{x}^{(\ell)}) \leq b$}{
    Set $x^{(\ell)} := \hat{x}^{(\ell)}$, $x := x+\hat{x}^{(\ell)}$; \tcp*[f]{online allocation}\\
  }
}
\end{algorithm}

We start with the proof of Theorem~\ref{theorem:LP_unknown}. The proof for Theorem~\ref{theorem:LP_known} is very similar and we mainly point out the differences.

Again we denote with $f(x \mid \hat{x}) = f(x\cup \hat{x}) - f(\hat{x})$ the contribution of $x$ to $\hat{x}$. Here, $(x\cup \hat{x})_j = \max\{x_j, x'_j\}$ is the component-wise maximum of $x$ and $\hat{x}$. Now, let $x_{\geq \ell}$ be the allocation by the online algorithm in rounds $\ell$ to $n$. Analogously, we denote the tentative allocation over the same period by $\hat{x}_{\geq \ell}$.

In contrast to the Section~\ref{sec:matching}, we need a Chernoff bound to lower bound the probability that the tentative allocation is feasible. This was also shown in~\cite{DBLP:conf/stoc/KesselheimTRV14}.

\begin{lemma}\label{lemma:LP_feasible}
For all $\ell \leq \frac{n}{4e\psi}$ with $\psi = d^{\frac{1}{B-1}}$ the probability that $\hat{x}_{\ell}$ is included in the final allocation is $\Pr{\sum_{\ell'<\ell} A \hat{x}^{(\ell')} \leq b_i-1 \growingmid \hat{x}^{(\ell)}, \ldots, \hat{x}^{(n)}} \geq \frac{1}{2}$.
\end{lemma}

\begin{proof}
Let $\mathcal{E}$ be any outcome for $\hat{x}^{(\ell + 1)}, \ldots, \hat{x}^{(n)}$. We have for all $i\in[m]$
\[\Ex{\left(\sum_{\ell'<\ell}A\hat{x}^{(\ell')}\right)_i \growingmid \mathcal{E}} \leq \sum_{\ell'=1}^{\ell}\frac{1}{\ell'}\frac{\ell'}{n}b_i=\frac{\ell-1}{n}b_i\leq \frac{b_i}{4e\psi}\enspace.\]
For $\delta = 4\e\psi \left(1 -  \frac{1}{b_i} \right)- 1$, we have $(1 + \delta) \frac{1}{4\e\psi} b_i = b_i - 1$.
At this point, we apply a Chernoff-bound and get
\begin{align*}\Pr{\left(\sum_{\ell' < \ell} A \hat{x}^{(\ell')}\right)_i \geq b_i - 1\growingmid \mathcal{E}} &= \Pr{\left(\sum_{\ell' < \ell} A \hat{x}^{(\ell')}\right)_i \geq (1 + \delta) \frac{b_i}{4\e\psi} \growingmid \mathcal{E}} \\
 &\leq \left( \frac{\e^\delta}{(1 + \delta)^{1 + \delta}} \right)^{\frac{b_i}{4\e\psi}} \leq \left( \frac{\e}{1 + \delta} \right)^{(1 + \delta) \frac{b_i}{4\e\psi}} \enspace .
\end{align*}
Please note here that the $\hat{x}^{(\ell)}$ are not independent. As discussed in~\cite{DBLP:conf/stoc/KesselheimTRV14}, we can still apply the bound since the randomization up to round $\ell$ is unbiased even conditioned on the outcomes of previous rounds. With $b_i \geq 2$, we have $1 + \delta \geq 2 \e \psi$ and therefore 
\[
\left( \frac{\e}{1 + \delta} \right)^{(1 + \delta) \frac{b_i}{4\e\psi}} 
\leq \left( \frac{1}{2\psi}  \right)^{b_i - 1} 
\leq \frac{1}{2d} \enspace .
\]
Since $d$ is the column sparsity, a union bound gives that the tentative assignment $\hat{x}_{\ell'}$ is carried out with probability at least $\frac{1}{2}$.
\end{proof}

To prove Theorem~\ref{theorem:LP_unknown}, we apply the technique from Lemma~\ref{lemma:matching_tentative} and bound the expected value $\Ex{f\left(x_{\ell}\growingmid x_{\geq\ell}\right)\growingmid S, \hat{x}_{\geq\ell}}$ gained in round $\ell$ conditioned on the columns that arrived earlier $S$ and the future tentative allocations $ \hat{x}_{\geq\ell}$. Next, we analyze the value collected by the algorithm recursively similar to Lemma~\ref{lemma:with-recursion}. It is important to note here, that our analysis only respects the value collected in rounds $\frac{n}{8e\psi} \leq \ell \leq \frac{n}{4e\psi}$. This is due to the fact that Lemma~\ref{lemma:LP_feasible} does not hold during the first and last rounds.

For Theorem~\ref{theorem:LP_known}, we change Algorithm~\ref{algo:lp} slightly and introduce a sampling phase. The modified algorithm samples a $p = 1-\frac{1}{2e}\left(\frac{1}{2d}\right)^{\frac{1}{B-1}}$-fraction of the input sequence and then behaves like Algorithm~\ref{algo:lp}. The proof requires a similar probability bound like Lemma~\ref{lemma:LP_feasible} for rounds $pn\leq\ell\leq n$. Then it is analogous to the proof of Theorem~\ref{theorem:LP_unknown}.

\begin{proof}[Proof (of Theorem~\ref{theorem:LP_unknown})]
Let us define
\[
x^{\geq\ell} = \sum_{\ell' = \ell}^{\frac{n}{4e\psi}} x^{(\ell')} \qquad \text{ and } \qquad \hat{x}^{\geq\ell} = \sum_{\ell' = \ell}^{\frac{n}{4e\psi}} \hat{x}^{(\ell')} \enspace.
\]

We use the technique from Lemma~\ref{lemma:matching_tentative} and bound 
\begin{equation}\label{eq:lp_conditional_expectation}
\begin{aligned}
\Ex{F\left(x^{(\ell)}\growingmid x^{\geq\ell+1}\right)\growingmid S, \hat{x}^{(\ell)}, \hat{x}^{\geq\ell}} &\geq \Ex{F\left(x^{(\ell)}\growingmid \hat{x}^{\geq\ell+1}\right)\growingmid S, \hat{x}^{(\ell)}, \hat{x}^{\geq\ell+1}}\\
&= \Ex{F\left(\hat{x}^{(\ell)}\growingmid \hat{x}^{\geq\ell+1}\right)\growingmid S, \hat{x}^{(\ell)}, \hat{x}^{\geq\ell+1}}\Pr{\sum_{\ell'<\ell} A \hat{x}^{(\ell')} \leq b_i-1}\\
&\geq \frac{1}{2\ell}F\left(\tilde{x}^{(\ell)}\growingmid \hat{x}^{\geq\ell+1}\right)\geq \frac{1}{2\ell}\left(F(\tilde{x}^{(\ell)}) - F(\hat{x}^{\geq\ell})\right)\enspace.
\end{aligned}
\end{equation}

Now we express the expected value collected by the algorithm recursively
\begin{equation}\label{eq:lp_recursion}
\begin{aligned}
\Ex{F(x^{\geq\ell})} &\geq \Ex{F(x^{(\ell)}\mid x^{\geq\ell+1}) + F(x^{\geq\ell+1})}\\
&\geq \frac{1}{2\ell}\Ex{F(\tilde{x}^{(\ell)})} - \frac{1}{2\ell} \Ex{F(\hat{x}^{\geq\ell+1})} + \Ex{F(x^{\geq\ell+1})}
\enspace.
\end{aligned}
\end{equation}
Using Lemma~\ref{lemma:LP_feasible} once again, we have $\Ex{F(x^{\geq\ell+1})} \geq \frac{1}{2} \Ex{F(\hat{x}^{\geq\ell+1})}$. Furthermore, we have $\Ex{F(\tilde{x}^{(\ell)})} \geq \alpha \frac{\ell^2}{n^2}F(\OPT)$ because the optimal solution in $\mathcal{P}(\frac{\ell}{n}, S)$ has expected value at least $\frac{\ell^2}{n^2} F(\OPT)$. This is due to the fact that every variable from $\OPT$ is included with probability $\frac{\ell}{n}$ and constraints are scaled by another factor $\frac{\ell}{n}$. The vector $\tilde{x}^{(\ell)}$ is an $\alpha$-approximation of this solution. In combination, we get
\[
\Ex{F(x^{\geq\ell})} \geq \frac{\alpha\ell}{2n^2}F(\OPT) + \left(1-\frac{1}{\ell}\right)\Ex{F(x^{\geq\ell+1})}\enspace.
\]

Solving the recursion yields the desired result
\begin{align*}
\frac{\Ex{F(x)}}{\alpha F(\OPT)} \geq \frac{\Ex{F(x^{\geq\frac{n}{8e\psi}})}}{\alpha F(\OPT)} &\geq \sum_{j=\frac{n}{8e\psi}}^{\frac{n}{4e\psi}}\prod_{i=\frac{n}{8e\psi}}^{j-1}\left(1-\frac{1}{i}\right)\frac{\alpha j}{2n^2} = \sum_{j=\frac{n}{8e\psi}}^{\frac{n}{4e\psi}} \frac{\frac{n}{8e\psi}}{j-1}\frac{\alpha j}{2n^2}\\
&\geq \left(\frac{n}{4e\psi} - \frac{n}{8e\psi}\right) \frac{1}{8e\psi}\frac{\alpha}{2n} \in \Omega\left(\alpha d^{-\frac{2}{B-1}}\right)\enspace.\qedhere
\end{align*}
\end{proof}

\begin{proof}[Proof (of Theorem~\ref{theorem:LP_known})]
This time we define
\[
x^{\geq\ell} = \sum_{\ell' = \ell}^n x^{(\ell')} \qquad \text{ and } \qquad \hat{x}^{\geq\ell} = \sum_{\ell' = \ell}^n \hat{x}^{(\ell')} \enspace.
\]

Analogous to Lemma~\ref{lemma:LP_feasible}, we bound the probability that the tentative allocation is feasible with a Chernoff bound. For $(1 + \delta) (1 - p) b_i = b_i - 1$ we have
\[
\Pr{\sum_{\ell'<\ell} A \hat{x}^{(\ell')} \leq b_i-1} \geq \frac{1}{2}
\]
because
\begin{align*}
\Pr{\left(\sum_{\ell' < \ell} A x^{(\ell')}\right)_i \geq b_i - 1} &= \Pr{\left(\sum_{\ell' < \ell} A x^{(\ell')}\right)_i \geq (1 + \delta) (1 - p) b_i} \leq \left( \frac{\e^\delta}{(1 + \delta)^{1 + \delta}} \right)^{(1-p) b_i} \\
&\leq \left( \frac{\e}{1 + \delta} \right)^{b_i - 1} = \left( \frac{\e (1-p)}{1 - \frac{1}{b_i}} \right)^{b_i - 1}\leq \left( 2 \e (1-p) \right)^{b_i - 1} = \frac{1}{2d}\enspace.
\end{align*}
Here we use the column sparsity in a union bound and get the desired success probability.

Just like in the previous proof for Theorem~\ref{theorem:LP_unknown}, we use Equation~\eqref{eq:lp_conditional_expectation} 
\[
\Ex{F\left(x^{(\ell)}\growingmid x^{\geq\ell}\right)\growingmid S, \hat{x}^{\geq\ell}} \geq \frac{1}{2\ell}\left(F(\tilde{x}^{(\ell)}) - F(\hat{x}^{\geq\ell})\right)
\]
and get the same recursion like in Equation~\eqref{eq:lp_recursion}
\[
\Ex{F(x^{\geq\ell})} \geq \Ex{F(x^{(\ell)}\mid x^{\geq\ell+1}) + F(x^{\geq\ell+1})} \enspace.
\]
We get the same recursion $\Ex{F(x^{\geq\ell})} \geq \Ex{\frac{1}{2\ell}\left(F(\tilde{x}^{(\ell)}) - F(\hat{x}^{\geq\ell+1})\right) + F(x^{\geq\ell+1})}$, but this time we sum over a different set of $\ell$. 
With $\Ex{F(\tilde{x}^{(\ell)})} \geq \alpha\frac{\ell^2}{n^2}F(\OPT)$, we have for $pn\leq\ell\leq n$

\[\Ex{F(x^{\geq\ell})} \geq \frac{\alpha\ell}{2n^2}F(\OPT) + \left(1-\frac{1}{\ell}\right)\Ex{F(x^{\geq\ell+1})}\enspace.\]

The recursion yields
\begin{align*}
\frac{\Ex{F(x^{\geq\ell})}}{F(\OPT)} &\geq \sum_{j=pn+1}^n\prod_{i=pn+1}^{j-1}\left(1-\frac{1}{i}\right)\frac{\alpha j}{2n^2} = \sum_{j=pn+1}^n \frac{pn}{j-1}\frac{\alpha j}{2n^2}\\
&\geq \alpha \left(n - pn\right)\frac{p}{2n} = \alpha (1-p)\frac{p}{2} \in \Omega\left(\alpha d^{-\frac{1}{B-1}}\right)\enspace.\qedhere
\end{align*}
\end{proof}

\bibliographystyle{plain}

\bibliography{main_document}

\begin{thebibliography}{10}

\bibitem{DBLP:conf/soda/AgrawalD15}
Shipra Agrawal and Nikhil~R. Devanur.
\newblock Fast algorithms for online stochastic convex programming.
\newblock In {\em Proc.\ 26th Symp.\ Discr.\ Algorithms (SODA)}, pages
  1405--1424, 2015.

\bibitem{DBLP:journals/ior/AgrawalWY14}
Shipra Agrawal, Zizhuo Wang, and Yinyu Ye.
\newblock A dynamic near-optimal algorithm for online linear programming.
\newblock {\em Operations Research}, 62(4):876--890, 2014.

\bibitem{DBLP:conf/soda/BabaioffIK07}
Moshe Babaioff, Nicole Immorlica, and Robert Kleinberg.
\newblock Matroids, secretary problems, and online mechanisms.
\newblock In {\em Proc.\ 18th Symp.\ Discr.\ Algorithms (SODA)}, pages
  434--443, 2007.

\bibitem{DBLP:journals/talg/BateniHZ13}
MohammadHossein Bateni, Mohammad~Taghi Hajiaghayi, and Morteza Zadimoghaddam.
\newblock Submodular secretary problem and extensions.
\newblock {\em ACM Trans.\ Algorithms}, 9(4):32, 2013.

\bibitem{DBLP:conf/soda/BuchbinderFNS14}
Niv Buchbinder, Moran Feldman, Joseph Naor, and Roy Schwartz.
\newblock Submodular maximization with cardinality constraints.
\newblock In {\em Proc.\ 25th Symp.\ Discr.\ Algorithms (SODA)}, pages
  1433--1452, 2014.

\bibitem{DBLP:journals/siamcomp/CalinescuCPV11}
Gruia C{\u{a}}linescu, Chandra Chekuri, Martin P{\'{a}}l, and Jan
  Vondr{\'{a}}k.
\newblock Maximizing a monotone submodular function subject to a matroid
  constraint.
\newblock {\em SIAM J. Comput.}, 40(6):1740--1766, 2011.

\bibitem{DBLP:conf/sigecom/DevenurH09}
Nikhil~R. Devenur and Thomas~P. Hayes.
\newblock The adwords problem: online keyword matching with budgeted bidders
  under random permutations.
\newblock In {\em Proc.\ 10th Conf.\ Econom.\ Comput.\ (EC)}, pages 71--78,
  2009.

\bibitem{DBLP:journals/corr/FeldmanI15}
Moran Feldman and Rani Izsak.
\newblock Building a good team: Secretary problems and the supermodular degree.
\newblock {\em CoRR}, abs/1507.06199, 2015.

\bibitem{DBLP:conf/approx/FeldmanNS11}
Moran Feldman, Joseph Naor, and Roy Schwartz.
\newblock Improved competitive ratios for submodular secretary problems
  (extended abstract).
\newblock In {\em Approximation, Randomization, and Combinatorial Optimization.
  Algorithms and Techniques - 14th International Workshop, {APPROX} 2011, and
  15th International Workshop, {RANDOM} 2011, Princeton, NJ, USA, August 17-19,
  2011. Proceedings}, pages 218--229, 2011.

\bibitem{DBLP:conf/esa/FeldmanNSW11}
Moran Feldman, Joseph Naor, Roy Schwartz, and Justin Ward.
\newblock Improved approximations for k-exchange systems - (extended abstract).
\newblock In {\em Proc.\ 19th European Symp.\ Algorithms (ESA)}, pages
  784--798, 2011.

\bibitem{DBLP:conf/soda/FeldmanSZ15}
Moran Feldman, Ola Svensson, and Rico Zenklusen.
\newblock A simple \emph{O}(log log(rank))-competitive algorithm for the
  matroid secretary problem.
\newblock In {\em Proc.\ 26th Symp.\ Discr.\ Algorithms (SODA)}, pages
  1189--1201, 2015.

\bibitem{DBLP:conf/focs/FeldmanZ15}
Moran Feldman and Rico Zenklusen.
\newblock The submodular secretary problem goes linear.
\newblock In {\em Proc.\ 56th Symp.\ Foundations of Computer Science (FOCS)},
  pages 486--505, 2015.

\bibitem{DBLP:conf/wine/GuptaRST10}
Anupam Gupta, Aaron Roth, Grant Schoenebeck, and Kunal Talwar.
\newblock Constrained non-monotone submodular maximization: Offline and
  secretary algorithms.
\newblock In {\em Proc.\ 6th Intl.\ Conf.\ Web and Internet Economics (WINE)},
  pages 246--257, 2010.

\bibitem{DBLP:conf/esa/KesselheimRTV13}
Thomas Kesselheim, Klaus Radke, Andreas T{\"{o}}nnis, and Berthold
  V{\"{o}}cking.
\newblock An optimal online algorithm for weighted bipartite matching and
  extensions to combinatorial auctions.
\newblock In {\em Proc.\ 21st European Symp.\ Algorithms (ESA)}, pages
  589--600, 2013.

\bibitem{DBLP:conf/stoc/KesselheimTRV14}
Thomas Kesselheim, Klaus Radke, Andreas T{\"{o}}nnis, and Berthold
  V{\"{o}}cking.
\newblock Primal beats dual on online packing lps in the random-order model.
\newblock In {\em Proc.\ 46th Symp.\ Theory of Computing (STOC)}, pages
  303--312, 2014.

\bibitem{DBLP:conf/soda/Kleinberg05}
Robert~D. Kleinberg.
\newblock A multiple-choice secretary algorithm with applications to online
  auctions.
\newblock In {\em Proc.\ 16th Symp.\ Discr.\ Algorithms (SODA)}, pages
  630--631, 2005.

\bibitem{DBLP:conf/stoc/KorulaMZ15}
Nitish Korula, Vahab~S. Mirrokni, and Morteza Zadimoghaddam.
\newblock Online submodular welfare maximization: Greedy beats 1/2 in random
  order.
\newblock In {\em Proc.\ 47th Symp.\ Theory of Computing (STOC)}, pages
  889--898, 2015.

\bibitem{DBLP:conf/icalp/KorulaP09}
Nitish Korula and Martin P{\'{a}}l.
\newblock Algorithms for secretary problems on graphs and hypergraphs.
\newblock In {\em Proc.\ 36th Intl.\ Coll.\ Autom.\ Lang.\ Program.\ (ICALP)},
  pages 508--520, 2009.

\bibitem{Krause12survey}
Andreas Krause and Daniel Gloving.
\newblock Submodular function maximization.
\newblock In {\em Tractability: Practical Approaches to Hard Problems},
  chapter~3. Cambridge University Press, 2014.

\bibitem{DBLP:journals/mor/KulikST13}
Ariel Kulik, Hadas Shachnai, and Tami Tamir.
\newblock Approximations for monotone and nonmonotone submodular maximization
  with knapsack constraints.
\newblock {\em Math. Oper. Res.}, 38(4):729--739, 2013.

\bibitem{DBLP:conf/focs/Lachish14}
Oded Lachish.
\newblock O(log log rank) competitive ratio for the matroid secretary problem.
\newblock In {\em Proc.\ 55th Symp.\ Foundations of Computer Science (FOCS)},
  pages 326--335, 2014.

\bibitem{DBLP:journals/mor/LeeSV10}
Jon Lee, Maxim Sviridenko, and Jan Vondr{\'{a}}k.
\newblock Submodular maximization over multiple matroids via generalized
  exchange properties.
\newblock {\em Math.\ Oper.\ Res.}, 35(4):795--806, 2010.

\bibitem{DBLP:journals/mst/Ma0W16}
Tengyu Ma, Bo~Tang, and Yajun Wang.
\newblock The simulated greedy algorithm for several submodular matroid
  secretary problems.
\newblock {\em Theoret.\ Comput.\ Sci.}, 58(4):681--706, 2016.

\bibitem{DBLP:journals/mor/MolinaroR14}
Marco Molinaro and R.~Ravi.
\newblock The geometry of online packing linear programs.
\newblock {\em Math.\ Oper.\ Res.}, 39(1):46--59, 2014.

\bibitem{NemhauserW78}
G.~L. Nemhauser and L.~A. Wolsey.
\newblock Best algorithms for approximating the maximum of a submodular set
  function.
\newblock {\em Math.\ Oper.\ Res.}, 3(3):177--188, 1978.

\bibitem{DBLP:journals/mp/NemhauserWF78}
George~L. Nemhauser, Laurence~A. Wolsey, and Marshall~L. Fisher.
\newblock An analysis of approximations for maximizing submodular set functions
  - {II}.
\newblock {\em Math.\ Prog.}, 14(1):265--294, 1978.

\bibitem{DBLP:journals/orl/Sviridenko04}
Maxim Sviridenko.
\newblock A note on maximizing a submodular set function subject to a knapsack
  constraint.
\newblock {\em Oper.\ Res.\ Lett.}, 32(1):41--43, 2004.

\end{thebibliography}
\end{document}